\documentclass[12pt]{article}

\usepackage{amsmath,amssymb}
\usepackage{amsthm}
\usepackage{mathtools}
\usepackage{dsfont}
\usepackage{graphicx}
\usepackage{braket}
\usepackage{booktabs}
\usepackage{hyperref} % For cross-referencing
\hypersetup{
    colorlinks=true,
    linkcolor=blue,
    urlcolor=blue,
    citecolor=blue,
}
\usepackage[numbers,sort&compress]{natbib}
\usepackage[capitalise]{cleveref}

\newtheorem{theorem}{Theorem}
\newtheorem{proposition}{Proposition}
\newtheorem{corollary}{Corollary}
\newtheorem{lemma}{Lemma}
\newtheorem{observation}{Observation}

\newtheorem{definition}{Definition}

\def\sss {\scriptscriptstyle}
\def\tr {\mathrm{Tr}}

\allowdisplaybreaks[3]

\title{General Schemes for Quantum Entanglement and Steering Detection}

\author
{Ma-Cheng Yang$^{1}$, Cong-Feng Qiao$^{1,2\ast}$\\
\\
\normalsize{$^{1}$School of Physical Sciences, University of Chinese Academy of Sciences}\\
\normalsize{YuQuan Road 19A, Beijing 100049, China}\\
\normalsize{$^{2}$Key Laboratory of Vacuum Physics, University of Chinese Academy of Sciences}\\
\normalsize{YuQuan Road 19A, Beijing 100049, China}\\
\\
\normalsize{$^\ast$To whom correspondence should be addressed; E-mail: qiaocf@ucas.ac.cn.}
}

\date{}

\begin{document} 


\baselineskip24pt
%\linenumbers

\maketitle 

\begin{abstract}
Separability problem is a long-standing tough issue in quantum information theory. In this paper, we propose a general method to detect entanglement via arbitrary measurement $\boldsymbol{X}$, by which several novel criteria are established. The new criteria are found incorporate many of the prevailing ones in the literatures. Our method is applicable as well to the steering detection, which possesses a merit of ignoring the knowledge of involved quantum states. A concept of measurement orbit, which plays an important role in the detection of entanglement and steering, is introduced, which enlightens our understanding of uncertainty relation. Moreover, an extension of symmetric informationally complete positive operator-valued measures (SIC-POVM), viz. symmetric complete measurements (SCM), is proposed, and employed to reconstruct the quantum state analytically.
\end{abstract}

\section{Introduction}\label{sec:intro}
\noindent
Since Einstein, Podolsky and Rosen questioned the completeness of quantum mechanics \cite{einstein35}, the enigmatic correlation between far apart quantum systems arouses more and more attentions with the time going. People now realize that there are several different types of nonclassical correlation, such as quantum entanglement \cite{horodecki09,guhne09}, Bell nonlocality \cite{brunner14}, quantum steering \cite{uola20} and so on. These correlations are normally associated with certain entangled states, and to ascertain the entanglement tends to be an NP-hard issue with the increase of dimension, even for the bipartite system \cite{gurvits04}.

In recent years, many methods were proposed to distinguish the entangled states from the unentangled states \cite{horodecki09,guhne09}. Nevertheless, these methods are either hard to operate, like positive map theory \cite{horodecki96-TXC1O}, or satisfy only one of the necessary and sufficient conditions for entanglement. It should be noted that the positive partial transpose (PPT) criterion \cite{peres96} is necessary and sufficient just for systems $2\times 2$ and $2\times 3$ \cite{horodecki96-TXC1O}. For high dimensional systems, the PPT entangled states \cite{horodecki97,bennett99} may result in the bound entanglement problem \cite{horodecki98}. To capture these ``weaker'' PPT entangled states, some new criteria have been proposed. We know that the uncertainty relation is closely related to the entanglement \cite{hofmann03} and especially the method by means of it (named local uncertainty relation criterion, LUR) is effective for the detection of some PPT entangled states \cite{hofmann03-aoEiI}. More generally, entanglement can be uncovered by a concave-function uncertainty relation \cite{huang10}. In addition, it is proved that covariance matrix criterion is equivalent to a LUR criterion \cite{guhne07,gittsovich08}. Quantum Fisher information (QFI) has also been employed to detect the entanglement, which complements the criterion based on variance and LUR \cite{li13}. The computable cross norm or realignment (CCNR) criterion \cite{chen03,rudolph05} is also an effective way of detecting PPT entangled states. However, it is not a necessary and sufficient condition, even for 2-qubit system \cite{rudolph03}. Furthermore, the correlation matrix criteria based on Bloch representation provides as well a class of effective methods \cite{vicente07,li18}. Recently, a entanglement criterion based on SIC-POVM is proposed \cite{shang18}, which may outperform the CCNR in some cases.

In comparison with the bipartite system, multipartite entanglement encounters more challenges due to the complexity of multipartite system. It is realized that the bipartite methods implies somehow the detection of the multipartite entanglement, and in fact various multipartite criteria have been successfully developed, such as the genuine multipartite entanglement criterion \cite{jungnitsch11}, quantum Fisher information method for multipartite system \cite{toth12,hyllus12,hong15,akbari-kourbolagh19,ren21}, correlation tensor criteria \cite{hassan08,de11} and so on \cite{wu00,krammer09,li15-genu,huber10,ketterer19}. Very recently, we notice that a necessary and sufficient criterion for $N$-qubit system is given \cite{wu22}. More detailed progresses about entanglement criteria can be found in reviews \cite{horodecki09,guhne09}.

Though Schrödinger proposed the concept of steering as early as 1935 \cite{schrodinger35,schrodinger36}, people pay less attention to it until Wiseman \emph{et al.} \cite{wiseman07} formulated the definition in 2007. Quantum steering is an intermediate quantum correlation between quantum entanglement and Bell nonlocality \cite{wiseman07}. Recently, many steering criteria have been developed, similar to entanglement detection, such as criteria via various uncertainty relations \cite{schneeloch13,pramanik14,zhen16,costa18,li21} and so on (see e.g. recent review \cite{uola20}). Notice, the phenomenon of quantum steering has been explicitly displayed by a series of miraculous experiments \cite{,saunders10,handchen12,smith12,wittmann12,sun14}.

In attempt to detect quantum correlations, such as quantum entanglement and steering, one needs first to determine a group of observables or measurements. The CCNR and ESIC criteria belong to the orthogonal measurement (OM) and SIC-POVM measurement, respectively. In fact, generally, one can employ arbitrary measurement to set up a nonlocal correlation criteron, in which each observable can be aligned along arbitrary direction on the Poincar{\'e} sphere. In this work, we are about to propose a general method to obtain entanglement and steering criteria for the arbitrary measurement. To this end, a concept of the measurement orbit will be introduced, which exhibits explicitly the entanglement and steering detection. We will prove the equivalence of correlation matrix criteria and optimal entanglement witness which can be found in the measurement orbit. In addition, it is demonstrated that the measurements lying in the same measurement orbit possess a common variance uncertainty relation, which results in the optimal LUR criterion in the measurement orbit.

This paper is organized as follows. In section 2 and 3, we develop some necessary concepts and tools for the establishment of entanglement and steering criteria. In section 4, we give out a general variance uncertainty relation in summation belonging to the whole measurement orbit $\mathcal{O}(\boldsymbol{X})$ of arbitrary measurement $\boldsymbol{X}$ and introduce a
class of measurement SCM for the qudit system. In section 5, entanglement and steering criteria are obtained for arbitrary measurement. The last section is devoted to a brief conclusion.

\section{Quantum observables and representation}\label{sec:poincare_sphere}
\noindent
Generally, physical observables are described as Hermitian operators in quantum mechanics. For qubit system, it is convenient, for instance, to parametrize traceless operators by unit vectors on the Poincar{\'e} sphere, i.e., $X=\vec{n}\cdot\vec{\sigma}$, where $\vec{\sigma}$ is the vector with components of Pauli matrices; $\vec{n}$ is the unit vector in Euclidean space of dimension 3. Similarly, for qudit system, an arbitrary observable $X$ can be parametrized as
\begin{align}
X = \vec{n}\cdot \boldsymbol{\Pi} = n_{0}\Pi_{0} + \hat{n}\cdot\vec{\pi} \; ,
\label{eq:obser_param}
\end{align}
where $\Pi_{0}=\sqrt{\frac{2}{d}}\mathds{1}$ and $\{\Pi_{\mu}=\pi_{\mu}\}_{\mu=1}^{d^2-1}$ is orthogonal basis of $\mathfrak{su}(d)$ Lie algebra, satisfying the following relations:
\begin{align}
&\pi_{\mu}=\pi_{\mu}^{\dagger}\; , \qquad \mathrm{Tr}[\pi_{\mu}]=0 \; , \notag \\
&\pi_{\mu}\pi_{\nu} = \frac{2}{d}\delta_{\mu\nu}\mathds{1} + \sum_{\gamma}(if_{\mu\nu \gamma} + g_{\mu\nu\gamma})\pi_{\gamma} \; ,
\label{eq:su_lie_alge_gene}
\end{align}
with $g_{\mu\nu \gamma}$ and $f_{\mu\nu \gamma}$ the completely symmetric and antisymmetry tensors respectively. It is easy to check that $\{\Pi_{\mu}\}_{\mu=0}^{d^2-1}$ satisfies the following orthogonal relation
\begin{align}
\tr[\Pi_{\mu}\Pi_{\nu}] = 2\delta_{\mu\nu} \; .
\end{align}
Obviously, $\vec{n}$ is a real $d^2$-dimensional vector which origins from the hermiticity of observable $X$ and the first component $n_{0}$ of $\vec{n}$ gives the trace of $X$, that is, $\tr[X]=\sqrt{2d}\,n_{0}$.

According to above, an observable can be decomposed into two parts, with trace and without trace, of which the traceless one constitutes a sphere with radius $\alpha$ in $(d^2-1)$-dimensional space. When we only consider the traceless observables, the sphere can be viewed as a generalized Poincar{\'e} sphere of qubit system.

\cref{eq:obser_param} establishes a relation between real vector $\vec{n}$ and observable $X$, that means an operation may map to a real vector. Correspondingly, we can define a measurement set $\mathcal{M}_{m}$:
\begin{align}
\mathcal{M}_{m} := \{\boldsymbol{X}=(X_{1},\cdots,X_{m})^{\mathrm{T}}|X_{\mu}=X_{\mu}^{\dagger}\} \; ,
\end{align}
where the measurement $\boldsymbol{X}$ means a group of operators $X_{1},X_{2},\cdots, X_{m}$.

A measurement is termed homogeneous if it is traceless, its components possess the same length $|\hat{n}|=\sqrt{\sum_{i=1}^{d^2-1}n_{i}^2}=\alpha$. Accordingly, we have
\begin{align}
\tr[X^2] = 2|\vec{n}|^2 = 2\alpha^2 + \frac{1}{d}\tr^2[X] \; .
\label{eq:tr_meas_squar}
\end{align}
Therefore, a homogeneous measurement corresponds to $m$ real vectors $\hat{n}_{i}$ of dimension $(d^2-1)$ with $|\hat{n}_{i}|=\alpha$, i.e.,
\begin{align}
\boldsymbol{X} \longrightarrow \operatorname{CP} := \Bigg(\hat{n}_{1},\hat{n}_{2},\cdots,\hat{n}_{m}\Bigg) \; .
\end{align}
Here, the ends of vectors $\hat{n}_{1},\hat{n}_{2},\cdots,\hat{n}_{m}$ can be viewed as vertexes of a convex polytope(CP) in $(d^2-1)$-dimensional space, which implies a measurement set is linked with a convex polytope. We then have:
\begin{proposition}
A traceless homogeneous measurement $\boldsymbol{X}$ corresponds to a convex polytope in $(d^2-1)$-dimensional space.
\end{proposition}

Moreover, we define the orbit of measurement set $\mathcal{M}_{m}$ under orthogonal group $O(m)$ as
\begin{align}
\mathcal{O}(\boldsymbol{X}) := \{O\boldsymbol{X}|O\in O(m)\} \; , \forall \boldsymbol{X}\in \mathcal{M}_{m} \; ,
\end{align}
which partitions $\mathcal{M}_{m}$ to disjoint equivalence
classes under the action of $O(m)$ \cite{dummit03}. According to the definition, two measurements are equivalent if $\boldsymbol{Y}=O\boldsymbol{X}$, or in other words they lie in the same orbit.

\section{Representation of quantum state}

\noindent
Based on the discussion in preceding section we now attempt to characterize a quantum state by measurement $\boldsymbol{X}$. Similar to observable, a quantum state, in terms of density matrix can be parameterized as
\begin{align}
\rho = \frac{1}{2}\vec{r}\cdot\boldsymbol{\Pi} = \frac{1}{d}\mathds{1} + \frac{1}{2}\hat{r}\cdot\vec{\pi} \; .
\label{eq:dens_matr_param}
\end{align}
Here, $r_{0}=\sqrt{2/d}$ and $|\vec{r}|\leq \sqrt{2}$ due to the normalization and the positive semidefiniteness of $\rho$, respectively. $\hat{r}$ is usually referred as the Bloch vector of quantum state $\rho$ in the literature. In light of quantum mechanics, the expectation of observable $X_{\mu}$ reads
\begin{align}
\braket{X_{\mu}} = \tr[\rho X_{\mu}] = \vec{n}_{\mu}\cdot \vec{r} \; .
\end{align}
Therefore, a given quantum state can be characterized by the real vector stemming from the measurement $\boldsymbol{X}$:
\begin{align}
\vec{x} := (\braket{X_{1}},\braket{X_{2}},\cdots)^{\mathrm{T}} \; ,
\end{align}
which inspires us to define a convex set of $\mathcal{B}(\boldsymbol{X})$ to characterize a set of quantum states illustrated in \cref{fig:measurement}.
\begin{figure}
\centering
\includegraphics[width=1\linewidth]{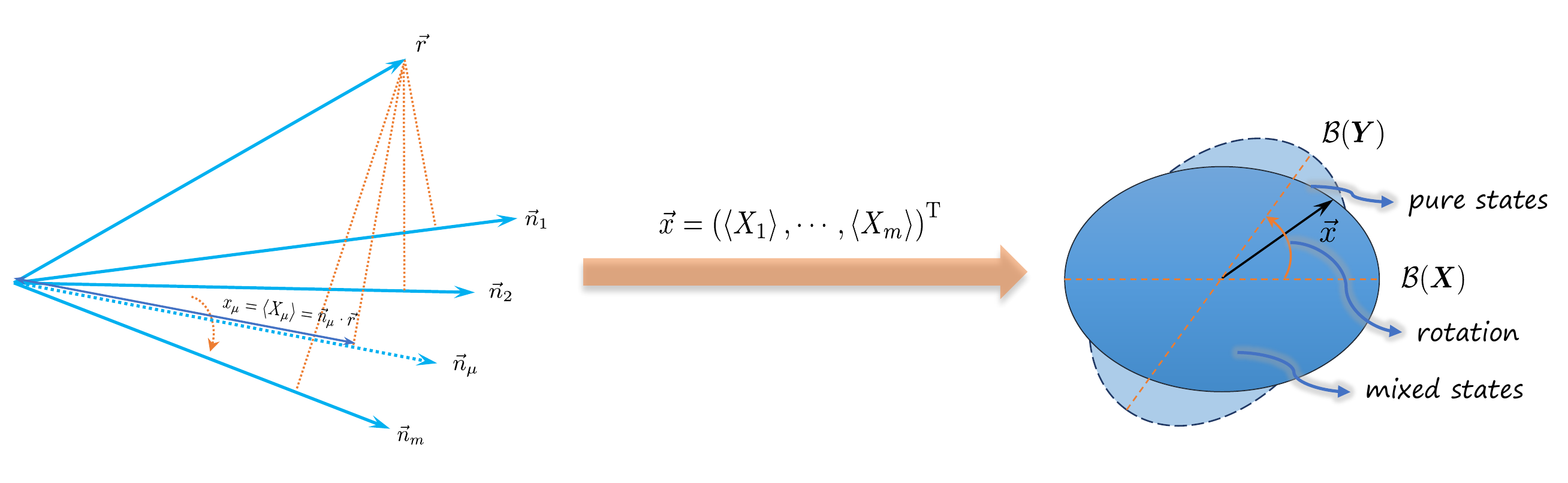}
\caption{Representation of quantum state under measurement $\boldsymbol{X}$. Every quantum state maps to a real vectors $\vec{x}$, whose set is denoted as $\mathcal{B}(\boldsymbol{X})$. The boundary of $\mathcal{B}(\boldsymbol{X})$ corresponds to the pure states with the mixed states inside. If measurements $\boldsymbol{X}$ and $\boldsymbol{Y}$ lie in the same orbit, $\mathcal{B}(\boldsymbol{X})$ and $\mathcal{B}(\boldsymbol{Y})$ are equivalent up to a rotation.}
\label{fig:measurement}
\end{figure}
\begin{definition}
Given a measurement $\boldsymbol{X}$, one can establish a mapping from density matrix $\rho$ to $m$-dimensional real vector $\vec{x}$, $\rho \xmapsto{x_{\mu}=\tr[\rho X_{\mu}]} \vec{x}$. The $\mathcal{B}(\boldsymbol{X})$ is employed to signify the set of various $\vec{x}$, i.e. $\mathcal{B}(\boldsymbol{X})=\{\vec{x}|x_{\mu}=\tr[\rho X_{\mu}],\forall \rho\}$.
\end{definition}
From the above definition, one can readily find $\mathcal{B}(\boldsymbol{X})$  possessing the following properties (see \cref{appen:state_rep_property} for proof).
\begin{proposition}
$\mathcal{B}(\boldsymbol{X})$ is a bounded closed convex set; $\mathcal{B}(\boldsymbol{X})$ is a convex hull consisting of vectors $\vec{x}$ corresponding to pure states, that is, $\mathcal{B}(\boldsymbol{X})=\mathrm{conv}(\{\vec{x}|x_{\mu}=\tr[\rho X_{\mu}],\rho^2=\rho,\forall \rho\})$;
$\forall \boldsymbol{X}\in \mathcal{M}_{m}$, measurements lying in the same orbit $\mathcal{O}(\boldsymbol{X})$ have a common $\mathcal{B}(\boldsymbol{X})$ up to a rotation in $m$-dimensional space. \label{prop:state_rep_property}
\end{proposition}
Generally, in order to expand a vector in $d$-dimensional space, we need at least $d$ linearly independent vectors. In this sense, one can expand the density matrix by measurement $\boldsymbol{X}$
\begin{align}
\rho = \sum_{\mu}\omega_{\mu}X_{\mu} \; .
\end{align}
Here, $\dim \operatorname{span}\{X_{\mu}\}_{\mu=1}^m=d^2$ and the coefficients $\omega_{\mu}$ will not be unique in case $m> d^2$. If density matrix can be constructed from measurement $\boldsymbol{X}$, the coefficient $\omega_{\mu}$ can then be expressed in linear combination of $x_{\mu}=\tr[\rho X_{\mu}]$, which is related to the Gram matrix
\begin{align}
\Omega_{\mu\nu} := \tr[X_{\mu}X_{\nu}] \; .
\end{align}
$\Omega$ is generally positive semidefinite and will be positive definite when $\{X_{\mu}\}$ are linearly independent \cite{horn12}. Because $x_{\mu}=\sum_{\nu}\omega_{\nu}\tr[X_{\mu}X_{\nu}]=\sum_{\nu}\Omega_{\mu\nu}\omega_{\nu}$, we can condense it to
\begin{align}
\vec{x} = \Omega \vec{\omega} \; .
\label{eq:linear_eqs}
\end{align}
Given a measurement $\boldsymbol{X}$, \cref{eq:linear_eqs} is a set of linear equations with respect to $\omega_{\mu}$, whose general solution can be obtained through the Moore-Penrose inverse $\Omega^{-}$ \cite{penrose55,ben-israel03}
\begin{align}
\vec{\omega} = \Omega^{-}\vec{x} + (\mathds{1}-\Omega^{-}\Omega)\vec{x}_{0} \;
\end{align}
with $\vec{x}_{0}$ an arbitrary vector. By dint of the properties of the Moore-Penrose inverse, we may have the following relation (see \cref{appen:purity_equality} for proof)
\begin{align}
\tr[\rho^2] = \vec{\omega}^{\mathrm{T}}\Omega\vec{\omega} = \vec{x}^{\mathrm{T}}\Omega^{-}\vec{x} \; ,
\label{eq:purity_equality}
\end{align}
which is irrelevant with the vector $\vec{x}_{0}$. Since $\Omega^{-}$ is a real symmetric matrix, it can be diagonalized by some orthogonal matrix $O$. Then, we have
\begin{align}
\tr[\rho^2] = \sum_{\mu}\lambda_{\mu}x_{\mu}'^2 \; .
\end{align}
Here, $\vec{x}\,'=O\vec{x}$ and $O\Omega^{-}O^{\mathrm{T}}=\operatorname{diag}\{\lambda_{1},\lambda_{2},\cdots,\lambda_{d^2}\}$. Hence, for linearly independent measurement $\{X_{\mu}\}$, $\mathcal{B}(\boldsymbol{X})$ satisfies constraint
\begin{align}
\frac{\tr[\rho^2]}{\lambda_{\max}} \leq |\vec{x}|^2 \leq \frac{\tr[\rho^2]}{\lambda_{\min}} \; .
\end{align}
Here, $\lambda_{\max}$ and $\lambda_{\min}$ are maxima and minima of $\{\lambda_{1},\lambda_{2},\cdots,\lambda_{d^2}\}$, respectively; and $\lambda_{i}>0$ due to the positive definiteness of $\Omega^{-}$\footnote{If $X_{\mu}$ are linearly independent, then $\Omega$ will be invertible and $\Omega^{-}=\Omega^{-1}$.}.

\section{Complete measurements and uncertainty relation}
\noindent
With the preceding discussion, we now come to the issue of uncertainty relation for arbitrary measurement $\boldsymbol{X}$ in form of variance summation.

The variance of observable $X$ on the quantum state $\rho$ is given by
\begin{align}
V(\rho,X) = \tr[\rho X^2] - \tr[\rho X]^2 \; .
\end{align}
We can therefore describe the uncertainty of a measurement $\boldsymbol{X}$ by
\begin{align}
\mathcal{V}(\rho,\boldsymbol{X}) := \sum_{\mu}V(\rho,X_{\mu}) \; .
\end{align}
One may readily find the following properties for $\mathcal{V}(\rho,\boldsymbol{X})$ (see \cref{appen:variance_sum_prop} for proof).
\begin{proposition}\label{prop:variance_sum_prop}
1. Variance of observable $X$ is independent of its trace, i.e.
\begin{align}
V(\rho,X=\vec{n}\cdot\boldsymbol{\Pi}) = V(\rho,X=\hat{n}\cdot\vec{\pi}) \; .
\end{align}

\noindent
2. The quantity $\mathcal{V}(\rho,\boldsymbol{X})$ is an invariant of the measurement orbit $\mathcal{O}(\boldsymbol{X})$, $\forall \boldsymbol{X}\in \mathcal{M}_{m}$.

\noindent
3. Additivity of $\mathcal{V}(\rho,\boldsymbol{X})$ exists for product states:
\begin{align}
\mathcal{V}\left(\bigotimes_{i}\rho^{(i)},\sum_{i}\boldsymbol{X}^{(i)}\right) = \sum_{i}\mathcal{V}(\rho^{(i)},\boldsymbol{X}^{(i)}) \; .
\end{align}
Here, $\boldsymbol{X}^{(i)}$ stands for the i'th subsystem of a comprehensive measurement.

\noindent
4. The quantity $\mathcal{V}(\rho,\boldsymbol{X})$ is a concave function of the convex set of density matrices, i.e.
\begin{align}
\mathcal{V}\left(\sum_{i}p_i\rho_i,\boldsymbol{X}\right) \geq \sum_{i}p_{i}\mathcal{V}\left(\rho_{i},\boldsymbol{X}\right) \;
\end{align}
with $p_i$ the probability of state being the $\rho_{i}$.
\end{proposition}
For an arbitrary measurement $\boldsymbol{X}$, one can calculate $\sum_{\mu} X_{\mu}^2$ as per \cref{eq:obser_param}:
\begin{align}
\sum_{\mu=1}^{m}X_{\mu}^2 = \frac{mh^2}{d^2}\mathds{1} + \frac{2h}{d}\sum_{\mu}\hat{n}_{\mu}\cdot\vec{\pi} + \mathcal{C}' \; .
\end{align}
Here, $\tr[X_{\mu}]=h$, $\mathcal{C}'=\sum_{\mu}(\hat{n}_{\mu}\cdot\vec{\pi})^2=\sum_{k,l}N_{kl}\pi_{k}\pi_{l}$ and $N=\sum_{\mu}\hat{n}_{\mu}\hat{n}_{\mu}^{\mathrm{T}}$. The first term is a constant matrix, the last two terms rely on the structure of the measurement convex polytope. $\mathcal{C}'$ is a homogeneous polynomial of generators $\pi_{\mu}$, which is just the quadratic Casimir operator of $\mathfrak{su}(d)$ Lie algebra. According to \cref{prop:variance_sum_prop}, for $\mathcal{V}(\rho,\boldsymbol{X})$ we only need to consider the traceless measurement $\boldsymbol{X}$, i.e. $h=0$, and
\begin{align}
\mathcal{V}(\rho,\boldsymbol{X}) = \braket{\mathcal{C}'} - |\vec{x}|^2  \; .
\label{eq:uncer_equality}
\end{align}
Hence, $\mathcal{V}(\rho,\boldsymbol{X})$ is a measurement orbit invariant. Note, \cref{eq:uncer_equality} holds for the whole measurement orbit $\mathcal{O}(\boldsymbol{X})$.

The above discussion may be summarized as the following theorem:
\begin{theorem}
Given a measurement $\boldsymbol{X}$, all measurements in the same measurement orbit satisfy a common state-independent uncertainty relation
\begin{align}
\mathcal{V}(\rho,\boldsymbol{X}) \geq \mathcal{S} \; , \; \mathcal{S}=\min_{\rho}(\braket{\mathcal{C}'} - |\vec{x}|^2) \; .
\label{eq:uncer_rela}
\end{align}
Here, $\mathcal{C}'=\sum_{k,l}N_{kl}\pi_{k}\pi_{l}$, $N=\sum_{\mu}\hat{n}_{\mu}\hat{n}_{\mu}^{\mathrm{T}}$ and $x_{\mu}=\tr[\rho X_{\mu}]$.
\end{theorem}
Furthermore, the bound $\mathcal{S}$ can be written as $\mathcal{S}=\min_{\rho}\braket{\mathcal{C}'} - \max_{\vec{x}\in\mathcal{B}(\boldsymbol{X})}|\vec{x}|^2$. For qubit system, one can choose Pauli matrices $\vec{\sigma}$ as Lie algebra base, $\mathcal{C}' = \sum_{\mu}(\hat{n}_{\mu}\cdot\vec{\sigma})^2 = \sum_{\mu}|\hat{n}_{\mu}|^2\mathds{1}$.
\begin{corollary}
For qubit system, given that a traceless homogeneous measurement $\boldsymbol{X}$, then all the measurements lying in its measurement orbit have the common state-independent uncertainty relation
\begin{align}
\mathcal{V}(\rho,\boldsymbol{X}) \geq \mathcal{S} \; , \; \mathcal{S} = m\alpha^2 - \max_{\vec{x}\in\mathcal{B}(\boldsymbol{X})}|\vec{x}|^2 \; .
\end{align}
Here, $\alpha=|\hat{n}_{\mu}|$ and $m$ is the component number of measurement $\boldsymbol{X}$ or the vertex number of the measurement convex polytope.
\end{corollary}
Because of $x_{\mu}=\tr[\rho X_{\mu}]=\hat{n}_{\mu}\cdot\hat{r}$, then we have
\begin{align}
\max_{\vec{x}\in\mathcal{B}(\boldsymbol{X})}|\vec{x}|^2 = \max_{\hat{r}}\sum_{\mu}(\hat{n}_{\mu}\cdot\hat{r})^2 \; .
\label{eq:qubit_uncer_rela_bound}
\end{align}
Here, $\hat{r}$ is Bloch vector. Because the Bloch space of qubit system is a $3$-dimensional unit sphere, it is ready to calculate \cref{eq:qubit_uncer_rela_bound}. Consider a dichotomy measurement $\Big(\hat{n}_{1},\hat{n}_{2}\Big)$, which has a angle $\theta$ between its components, we have
\begin{align}
\mathcal{S} = 2 - \max_{\hat{r}}((\hat{n}_{1}\cdot\hat{r})^2+(\hat{n}_{2}\cdot\hat{r})^2) \;
\end{align}
with $|\hat{n}_{1}|=|\hat{n}_{2}|=\alpha=1$. Without loss of generality, one can fix measurement in $x-y$ plane and take $\hat{n}_{1}=(1,0,0)^{\mathrm{T}},\hat{n}_{2}=(\cos\theta,\sin\theta,0)^{\mathrm{T}}$. It is enough to only consider Bloch vector in $x-y$ plane and one can take $\hat{r}=(\cos\delta,\sin\delta,0)^{\mathrm{T}}$, then
\begin{align}
\mathcal{S} &= 2 - \max_{\delta} (\cos^2\delta+\cos^2(\delta-\theta))  \\
&= 1 - |\cos\theta| \; .
\end{align}
It can be proved that the extreme value can be achieved in $\delta=\frac{\theta}{2} + n\pi,\theta\in[0,\pi/2]$ and $\delta=\frac{\pi-\theta}{2}+n\pi,\delta\in[\pi/2,\pi]$ respectively. Thus, for qubit system, a dichotomy measurement with a angle $\theta$ between its components satisfies the state-independent uncertainty relation
\begin{align}
\mathcal{V}(\rho,\boldsymbol{X}) \geq 1 - |\cos\theta| \; ,
\label{eq:ex_sum_var_relation}
\end{align}
which reproduces the result of Ref. \cite{busch14-uncer}. It is worth noting that the variance summation uncertainty relation has a well-defined state-independent bound $\mathcal{S}$, which still holds in mixed form \cite{Li-Qiao-2015,abbott16}. However, we can demonstrate there is no state-independent bound in variance product situation. Taking qubit system as an example, according to Robertson uncertainty relation \cite{robertson29}, a dichotomy measurement with angle $\theta$ between its components satisfies uncertainty relation
\begin{align}
V(\rho,X_{1})V(\rho,X_{2}) \geq \left|\frac{1}{2}\braket{\left[X_1,X_2\right]}\right|^2 \; .
\label{eq:robertson_uncer_relation}
\end{align}
Under the Bloch representation, the commutator between two observables can be expressed as
\begin{align}
[X_{1},X_{2}] = 2i(\hat{n}_{1}\times\hat{n}_{2})\cdot\vec{\sigma} \; .
\end{align}
Here, we make use of the commutation relation $[\sigma_{\mu},\sigma_{\nu}]=2i\sum_{\gamma}\epsilon_{\mu\nu\gamma}\sigma_{\gamma}$ and $\epsilon_{\mu\nu\gamma}$ signifies the $3$-order completely antisymmetric tensor. Thus, we have
\begin{align}
\left|\braket{[X_{1},X_{2}]}\right| = 4|(\hat{n}_{1}\times\hat{n}_{2})\cdot\hat{r}| \; .
\end{align}
Uncertainty relation \cref{eq:robertson_uncer_relation} can be reformulated as Bloch form
\begin{align}
V(\rho,X_{1})V(\rho,X_{2}) \geq 4|(\hat{n}_{1}\times\hat{n}_{2})\cdot\hat{r}|^2 \; .
\end{align}
Obviously if Bloch vector $\hat{r}$ lies in the plane of $\hat{n}_{1},\hat{n}_{2}$, the right hand side is identically zero, and the left-hand side hence has no nontrivial state-independent bound.

Furthermore, considering anticommutator relation $\{\sigma_{\mu},\sigma_{\nu}\}=2\delta_{\mu\nu}\mathds{1}$ and hence $\braket{\{X_{1},X_{2}\}}=2\hat{n}_{1}\cdot\hat{n}_{2}$, the Schrödinger type uncertainty relation \cite{schrodinger30}
\begin{align}
V(\rho,X_{1})V(\rho,X_{2}) \geq \left|\frac{1}{2}\braket{[X_{1},X_{2}]}\right|^2 + \left|\frac{1}{2}\braket{\{X_{1},X_{2}\}} - \braket{X_{1}}\braket{X_{2}}\right|^2 \;
\end{align}
can then be expressed as
\begin{align}
V(\rho,X_{1})V(\rho,X_{2}) \geq 4|(\hat{n}_{1}\times\hat{n}_{2})\cdot\hat{r}|^2 + \left|\hat{n}_{1}\cdot\hat{n}_{2} - (\hat{n}_{1}\cdot\hat{r})(\hat{n}_{2}\cdot\hat{r})\right|^2 \; .
\end{align}
If Bloch vector $\hat{r}$ is perpendicular to the plane of measurement $\hat{n}_{1},\hat{n}_{2}$, then
\begin{align}
V(\rho,X_{1})V(\rho,X_{2}) \geq 1 + 3\sin^2\theta \; .
\end{align}
And if Bloch vector $\hat{r}$ is in the plane of the measurement, without loss of generality, one can fix measurement in $x-y$ plane and take $\hat{n}_{1}=(1,0,0)^{\mathrm{T}},\hat{n}_{2}=(\cos\theta,\sin\theta,0)^{\mathrm{T}}$ and Bloch vector $\hat{r}=(\cos\delta,\sin\delta,0)^{\mathrm{T}}$, then
\begin{align}
V(\rho,X_{1})V(\rho,X_{2}) \geq \left|\cos\theta-\cos\delta\cos(\delta-\theta)\right|^2
\end{align}
If $\delta=\theta$, the right hand side still vanishes identically. It is concluded that the variance product uncertainty relation has no nontrivial state-independent bound.

Generally, for a qudit system, it is not readily to obtain $\mathcal{S}$. $\mathcal{C}'$ depends on structure of the measurement convex polytope. The highly symmetrical measurement will definitely simplify calculations. For the simplest case, the measurement convex polytope just is a regular $(d^2-1)$-simplex. Thus, we define the following measurement.
\begin{definition}
For a $d$-dimensional quantum system, there always exist the following symmetrical complete measurement (SCM)
\begin{align}
\{X_{\mu}=\vec{n}_{\mu}\cdot\boldsymbol{\Pi}|\cos\braket{\hat{n}_{\mu},\hat{n}_{\nu}}=\frac{-1}{d^2-1},\mu\neq\nu\}_{\mu=1}^{d^2} \; ,
\end{align}
where $\tr[X_{\mu}]=h$, $|\hat{n}_{\mu}|=\alpha$ and $\{\hat{n}_{\mu}\}$ constitutes a regular $(d^2-1)$-simplex geometrically.
\end{definition}
Obviously, the symmetric informationally complete positive operator-valued measures (SIC-POVM) \cite{renes04,graydon16} is a special case of SCM. Based on above definition, one can conclude the following properties of SCM (see \cref{appen:scm_prope} for proof).
\begin{proposition} \label{prop:scm_prope}
Elements of SCM satisfy the following properties
\begin{align}
\sum_{\mu=1}^{d^2}X_{\mu} = dh\mathds{1}& \; , \; \sum_{\mu=1}^{d^2}X_{\mu}^2 = (h^2 + 2d\alpha^2)\mathds{1} \; , \\
\tr[X_{\mu}X_{\nu}] =& \frac{h^2}{d} + \frac{2\alpha^2(d^2\delta_{\mu\nu}-1)}{d^2-1} \; .
\end{align}
And if $\alpha^2=\frac{d-1}{2d^3},h=\frac{1}{d}$ and $X_{\mu}\geq 0$, SCM will be an SIC-POVM.
\end{proposition}
Due to the symmetry of SCM, we can reconstruct quantum state by SCM and calculate $|\vec{x}|^2$ analytically, which are given in the following theorem (see \cref{appen:state_recon_scm} for proof).
\begin{theorem} \label{th:state_recon_scm}
Quantum state $\rho$ can be reconstructed as by an SCM
\begin{align}
\rho = \frac{d^2-1}{2d^2\alpha^2}\sum_{\mu}x_{\mu}X_{\mu} + \left(\frac{1}{d} - \frac{h^2(d^2-1)}{2d^2\alpha^2}\right)\mathds{1} \; ,
\end{align}
where $x_{\mu}=\tr[\rho X_{\mu}]$ and $\alpha,h>0$. And the bound of $\mathcal{B}(\boldsymbol{X})$ is related to the purity $\tr[\rho^2]$ of quantum state
\begin{align}
|\vec{x}|^2 = \frac{2d\alpha^2\left(d\tr[\rho^2]-1\right)}{d^2-1} + h^2 \; .
\end{align}
Here, $\alpha>0,h\geq 0$.
\end{theorem}
Simultaneously, we reach the following uncertainty relation equality for SCM.
\begin{corollary}
SCM measurement satisfies uncertainty relation equality
\begin{align}
\mathcal{V}(\rho,\boldsymbol{X}) = \frac{2d^2\left(d-\tr[\rho^2]\right)}{d^2-1}\alpha^2 \; .
\label{eq:scm_uncer_rela}
\end{align}
Here, $d$ denotes the dimension of Hilbert space.
\end{corollary}
The parameter $\alpha$ is not essential and we can take $\alpha=1$. \cref{eq:scm_uncer_rela} stems from the structure of SCM. Considering that the purity $1/d\leq\tr[\rho^2]\leq 1$, the state-independent uncertainty relation follows
\begin{align}
\frac{2d^2}{d+1} \leq \mathcal{V}(\rho,\boldsymbol{X}) \leq 2d \; .
\end{align}

\section{Quantum entanglement and steering detection}\label{sec:ent_det}

\subsection{Entanglement detection}

\noindent
Next, we discuss the entanglement detection in bipartite quantum system equipped with those concepts and tools developed in above sections. A bipartite state $\rho$, lies in Hilbert space $H_{A}\otimes H_{B}$, is entangled if it cannot be decomposed into \cite{werner89}
\begin{align}
\rho = \sum_{k}p_{k}\ket{\psi_{k}}\bra{\psi_{k}}\otimes \ket{\phi_{k}}\bra{\phi_{k}} \; .
\label{eq:sep_den_mat}
\end{align}
Here, $\ket{\psi_{k}}\in H_{A}$, $\ket{\phi_{k}}\in H_{B}$ and $\sum_{k}p_{k}=1,p_{k}>0$. Consider arbitrary measurements $\boldsymbol{X}^{A}$ and $\boldsymbol{X}^{B}$ on systems A and B respectively, one may define the following correlation matrices
\begin{align}
&\mathcal{C}_{\mu\nu} := \braket{X_{\mu}^{A}\otimes X_{\nu}^{B}} \; , \\
&\gamma_{\mu\nu} := \braket{X_{\mu}^{A}}\braket{X_{\nu}^{B}} - \mathcal{C}_{\mu\nu} \; .
\end{align}
Here, $\braket{X_{\mu}^{A}}=\tr[\rho_{A}X_{\mu}^{A}]$, and the reduced density matrix $\rho_{A}=\tr_{B}[\rho]$, similarly for $B$. In general, one can expand a bipartite density matrix by $\Pi_{\mu}\otimes\Pi_{\nu}$, i.e.,
\begin{align}
\rho=\frac{1}{4}\sum_{\mu\nu}\chi_{\mu\nu}\Pi_{\mu}\otimes\Pi_{\nu} \; .
\label{eq:bipartite_state}
\end{align}
Here, the coefficient matrix $\chi_{\mu\nu}=\tr[\rho(\Pi_{\mu}\otimes\Pi_{\nu})]$ with $\chi_{00}=\frac{2}{d}$ due to $\tr[\rho]=1$, which characterizes the bipartite quantum state. $\mathcal{C}$ and $\gamma$ imply the correlation information of the quantum state corresponding to the implemented measurements. It is found that $\mathcal{C},\gamma$ and $\chi$ satisfy the following relationship (see \cref{appen:corre_matrix} for proof)
\begin{align}
&\mathcal{C} = M_{A}^\mathrm{T}\chi M_{B} \; , \label{eq:corre_matrix_c} \\
&\gamma = M_{A}^\mathrm{T}(\chi'-\chi)M_{B} \; .
\label{eq:corre_matrix_gamma}
\end{align}
Here, $M_{A}=(\vec{n}_{1}^{A},\vec{n}_{2}^{A},\cdots,\vec{n}_{m}^{A})$, $M_{B}=(\vec{n}_{1}^{B},\vec{n}_{2}^{B},\cdots,\vec{n}_{m}^{B})$ correspond to the implemented measurements in the two subsystems respectively and $\chi'_{\mu\nu}=\tr[(\rho_{A}\otimes\rho_{B})(\Pi_{\mu}\otimes\Pi_{\nu})]=\frac{d}{2}\chi_{\mu 0}\chi_{0\nu}$. This relationship can be generalized to $N$-partite system
\begin{align}
\mathcal{C}_{\mu_{1}\mu_{2}\cdots\mu_{N}} = \sum_{\nu_{1}\nu_{2}\cdots\nu_{N}}n_{\mu_{1}\nu_{1}}^{(1)}n_{\mu_{2}\nu_{2}}^{(2)}\cdots n_{\mu_{N}\nu_{N}}^{(N)}\chi_{\nu_{1}\nu_{2}\cdots\nu_{N}} \; ,
\end{align}
where $\mathcal{C}$ and $\chi$ are $N$ order tensor and $\vec{n}_{\mu_{i}}^{(i)}$ are measurements corresponding to i'th subsystem and similarly for $\gamma$. Furthermore, it is found that the trace norm of the correlation matrices $\mathcal{C},\gamma$ are invariant in the measurement orbit (seeing \cref{appen:measu_orbit_invar} for proof).
\begin{observation}\label{ob:measu_orbit_invar}
$\|\mathcal{C}\|_{\mathrm{tr}}$ and $\|\gamma\|_{\mathrm{tr}}$ are measurement orbit invariants, where $\|X\|_{\mathrm{tr}}$ denotes trace norm i.e., the sum of singular values of matrix $X$.
\end{observation}
In Section 2, we established a map between quantum state and a bounded closed convex set $\mathcal{B}(\boldsymbol{X})$, which implies a separable criterion for the arbitrary measurements (see \cref{appen:corre_c_sep_cri} for proof).
\begin{theorem}\label{th:corre_c_sep_cri}
For arbitrary measurements $\boldsymbol{X}^{A}$ and $\boldsymbol{X}^{B}$, the correlation matrix $\mathcal{C}$ of the separable states satisfies
\begin{align}
\|\mathcal{C}\|_{\mathrm{tr}} \leq \kappa \; ,
\end{align}
otherwise it is an entangled state. Here, $\kappa=\kappa_{A}\kappa_{B}$ and $\kappa_{A}=\max_{\vec{x}^{A}\in\mathcal{B}(\boldsymbol{X}^{A})}|\vec{x}^{A}|,\kappa_{B}=\max_{\vec{x}^{B}\in\mathcal{B}(\boldsymbol{X}^{B})}|\vec{x}^{B}|$.
\label{th:univ_sep_cri}
\end{theorem}
Considering that \cref{ob:measu_orbit_invar} and \cref{prop:state_rep_property}, any pairs of measurements in the same measurement orbit leads to the same criteria, which is general result of Proposition 3 in Ref. \cite{shang18}.
In the \cref{tab:various_measurement,tab:scm}, we list the value $\max_{\vec{x}\in\mathcal{B}(\boldsymbol{X})}|\vec{x}|$ for SCM and some measurements occurred in the references, which recovers separability criteria in these references and derives new separability criteria based on SCM.
\begin{corollary}
For measurement SCM, the correlation matrix $\mathcal{C}$ of the separable states satisfies
\begin{align}
\|\mathcal{C}\|_{\mathrm{tr}} \leq \frac{2d\alpha^2}{d+1}+h^2 \; ,
\end{align}
and if $\alpha^2=\frac{d-1}{2d^3},h=\frac{1}{d}$, ESIC criterion \cite{shang18} can be obtained.
\end{corollary}
\cref{th:univ_sep_cri} also implies how to construct entanglement witness $\mathcal{W}$ \cite{horodecki01,guhne09} for arbitrary measurements $\boldsymbol{X}^{A}$ and $\boldsymbol{X}^{B}$. The following matrix is an entanglement witness
\begin{align}
\mathcal{W} = \kappa\mathds{1} - \sum_{\mu}X_{\mu}^{A}\otimes X_{\mu}^{B} \; . \label{eq:ent_witness}
\end{align}
And there exist an optimal entanglement witness in the measurement orbits $\mathcal{O}(\boldsymbol{X}^{A})$ and $\mathcal{O}(\boldsymbol{X}^{B})$
\begin{align}
\mathcal{W}_{O} = \kappa\mathds{1} - \sum_{\mu}\widetilde{X}_{\mu}^{A}\otimes \widetilde{X}_{\mu}^{B} \; .
\label{eq:ent_witness_opt}
\end{align}
Here, $\boldsymbol{\widetilde{X}}^{A}=(O^{A})^{\mathrm{T}}\boldsymbol{X}^{A}$, $\boldsymbol{\widetilde{X}}^{B}=(O^{B})^{\mathrm{T}}\boldsymbol{X}^{B}$ and $O^{A},O^{B}$ correspond to a singular value decomposition of $\mathcal{C}$, i.e. $\mathcal{C}=O^{A}\Sigma (O^{B})^{\mathrm{T}}$.
The entanglement witness $\mathcal{W}_{O}$ is equivalent to \cref{th:univ_sep_cri} (see \cref{appen:ent_witness} for proof).
\begin{table}
\caption{\label{tab:various_measurement}The values $\max_{\vec{x}\in\mathcal{B}(\boldsymbol{X})}|\vec{x}|$ of four classes of measurements. Measurement $\{\frac{h}{\sqrt{d}}\mathds{1},\frac{\pi_{\mu}}{\sqrt{2}}\}$ corresponds to result in Ref. \cite{sarbicki20} and the orthogonal measurement (OM) $\{X_{\mu}|\tr[X_{\mu}X_{\nu}]=\delta_{\mu\nu}\}$ corresponds to the well-known CCNR criterion \cite{rudolph03,chen03} and measurement $\{\pi_{\mu}\}$ recovers de Vicente's correlation matrix criterion \cite{vicente07}.}
\begin{tabular*}{\hsize}{@{}@{\extracolsep{\fill}}ccccc@{}}
\hline\hline   
Measurement & SCM & $\{\frac{h}{\sqrt{d}}\mathds{1},\frac{\pi_{\mu}}{\sqrt{2}}\}$ & $\{\pi_{\mu}\}$ & OM \\ \midrule
$\displaystyle\max_{\vec{x}\in\mathcal{B}(\boldsymbol{X})}|\vec{x}|$ & $\sqrt{\frac{2d\alpha^2}{d+1}+h^2}$ & $\sqrt{\frac{d-1+h^2}{d}}$ & $\sqrt{\frac{2(d-1)}{d}}$ & $1$ \\
\hline\hline 
\end{tabular*}
\end{table}
\begin{table}
\caption{\label{tab:scm}SIC-POVM as a special case of SCM. If $\alpha^2=\frac{d-1}{2d^3},h=\frac{1}{d}$, we recover the separability criteria in Ref. \cite{shang18} and this criterion is independent on the existence of SIC-POVM which agrees with Ref. \cite{sarbicki20}.}
\begin{tabular*}{\hsize}{@{}@{\extracolsep{\fill}}ccc@{}}
\hline\hline   
SCM & $(\alpha^2=\frac{d-1}{2d^3},h=\frac{1}{d})$ & $h=0$   \\ \midrule
$\displaystyle\max_{\vec{x}\in\mathcal{B}(\boldsymbol{X})}|\vec{x}|$ & $\sqrt{\frac{2}{d(d+1)}}$ & $\sqrt{\frac{2d\alpha^2}{d+1}}$  \\
\hline\hline 
\end{tabular*}
\end{table}
On the other hand, we notice that the sum of the diagonal elements of $\gamma$ is related to the quantities $\mathcal{V}(\rho,\boldsymbol{X}^{A}+\boldsymbol{X}^{B})$, and $\mathcal{V}(\rho_{A},\boldsymbol{X}^{A}),\mathcal{V}(\rho_{B},\boldsymbol{X}^{B})$
\begin{align}
2\sum_{\mu}\gamma_{\mu\mu} = \left(\mathcal{V}(\rho_{A},\boldsymbol{X}^{A})+\mathcal{V}(\rho_{B},\boldsymbol{X}^{B})\right) - \mathcal{V}(\rho,\boldsymbol{X}^{A}+\boldsymbol{X}^{B}) \; .
\label{eq:sep_gamma}
\end{align}
Here, $\boldsymbol{X}^{A}+\boldsymbol{X}^{B}$ is a measurement of bipartite, i.e. $\left(\boldsymbol{X}^{A}+\boldsymbol{X}^{B}\right)_{\mu} = X_{\mu}^{A}\otimes\mathds{1}+\mathds{1}\otimes X_{\mu}^{B}$. It is interesting that \cref{eq:sep_gamma} is very similar to quantum mutual information \cite{nielsen10} in form. More importantly, \cref{eq:sep_gamma} implies a stronger separability condition. Via \cref{prop:variance_sum_prop} and the concavity of variance, we have
\begin{align}
\mathcal{V}(\rho_{\mathrm{sep}},\boldsymbol{X}^{A}+\boldsymbol{X}^{B}) \geq \mathcal{S}_{A} + \mathcal{S}_{B} \; .
\label{eq:lur_crit}
\end{align}
Here, $\mathcal{S}_{A}, \mathcal{S}_{B}$ are the optimal bound of uncertainty relation of subsystems, that is,
\begin{align}
\mathcal{S}_{A}=\min_{\rho^{A}}\mathcal{V}(\rho^{A},\boldsymbol{X}^{A}) \; , \; \mathcal{S}_{B}=\min_{\rho^{B}}\mathcal{V}(\rho^{B},\boldsymbol{X}^{B}) \; .
\end{align}
\cref{eq:lur_crit} is called as local uncertainty relation (LUR) criterion \cite{hofmann03,guhne04-frAQB}. However, LUR criterion encounters obstacle of how to find the optimal measurement in practice \cite{guhne04-frAQB}, which is an uncomfortable shortage of that. From the perspective of the optimization, this optimization problem can be expressed as
\begin{align}
&\text{minimize} \quad \mathcal{V}(\rho,\boldsymbol{X}^{A}+\boldsymbol{X}^{B}) - \mathcal{S}_{A} - \mathcal{S}_{B} \notag \\
&\text{subject to} \quad \boldsymbol{X}^{A}\in \mathcal{M}_{m}^{A}\; , \; \boldsymbol{X}^{B}\in \mathcal{M}_{m}^{B} \; .
\end{align}
To detect entanglement of quantum state $\rho$, we should accomplish the exhaustive search for all the measurements. Nevertheless, it is not necessary to exhaust all the measurements but only all the measurement orbits due to invariance of $\mathcal{V}(\rho,\boldsymbol{X})$ in the measurement orbits. Thus we conclude the following more effective LUR criterion.
\begin{theorem}
For arbitrary measurements $\boldsymbol{X}^{A}$ and $\boldsymbol{X}^{B}$, the separable states satisfy the following local uncertainty relation
\begin{align}
\min_{\boldsymbol{X}^{A}\in\mathcal{O}(\boldsymbol{X}^A),\boldsymbol{X}^{B}\in\mathcal{O}(\boldsymbol{X}^B)}\mathcal{V}(\rho,\boldsymbol{X}^{A}+\boldsymbol{X}^{B}) \geq \mathcal{S}_{A} + \mathcal{S}_{B} \; .
\end{align}
otherwise it is an entangled state. Here, $\mathcal{S}_{A}, \mathcal{S}_{B}$ are the optimal bound of uncertainty relation of subsystems, that is, $\mathcal{S}_{A}=\min_{\rho^{A}}\mathcal{V}(\rho^{A},\boldsymbol{X}^{A}),\mathcal{S}_{B}=\min_{\rho^{B}}\mathcal{V}(\rho^{B},\boldsymbol{X}^{B})$.
\label{th:opt_lur_crit}
\end{theorem}
\cref{th:opt_lur_crit} can be reformulated by correlation matrix $\gamma$. In light of \cref{eq:sep_gamma} and the orbit invariance of $\mathcal{V}(\rho,\boldsymbol{X})$, we have
\begin{align}
\min\mathcal{V}(\rho,\boldsymbol{X}^{A}+\boldsymbol{X}^{B}) = &\left(\mathcal{V}(\rho_{A},\boldsymbol{X}^{A})+\mathcal{V}(\rho_{B},\boldsymbol{X}^{B})\right) - \notag \\
&2\max\sum_{\mu}\gamma_{\mu\mu} \; .
\end{align}
In the optimal case, we obtain a tidy result $\max\sum_{\mu}\gamma_{\mu\mu}=\|\gamma\|_{\mathrm{tr}}$ (seeing \cref{appen:gamma_orbit_opt} for proof). Thus, we conclude the following separability condition equivalent to \cref{th:opt_lur_crit}.
\begin{theorem}
For arbitrary measurements $\boldsymbol{X}^{A}$ and $\boldsymbol{X}^{B}$, the correlation matrix $\gamma$ of the separable states satisfies
\begin{align}
\|\gamma\|_{\mathrm{tr}} \leq \frac{\left(\mathcal{V}(\rho_{A},\boldsymbol{X}^{A})-\mathcal{S}_{A}\right) + \left(\mathcal{V}(\rho_{B},\boldsymbol{X}^{B})-\mathcal{S}_{B}\right) }{2} \; .
\end{align}
otherwise it is an entangled state.
\label{th:univ_sep_cri_gamma}
\end{theorem}
\begin{table}
\caption{\label{tab:various_measurement_vs}The values $\mathcal{S}$ and $\mathcal{V}(\rho,\boldsymbol{X})$ for the aforementioned measurements. If we take the orthogonal measurement $\{X_{\mu}|\tr[X_{\mu}X_{\nu}]=\delta_{\mu\nu}\}$, we have $\|\gamma\|_{\mathrm{tr}}\leq \frac{2-(\tr[\rho_{A}^2]+\tr[\rho_{B}^2])}{2}$ for separable states, which is Theorem 2 in Ref. \cite{zhang07}.}
\begin{tabular*}{\hsize}{@{}@{\extracolsep{\fill}}cccc@{}}
\hline\hline   
Measurement & SCM & $\{\frac{h}{\sqrt{d}}\mathds{1},\frac{\pi_{\mu}}{\sqrt{2}}\}$ & OM \\ \midrule
$\mathcal{S}$ & $\frac{2d^2\alpha^2}{d+1}$ & $d-1$ & $d-1$ \\
$\mathcal{V}(\rho,\boldsymbol{X})$ & $\frac{2d^2\left(d-\tr[\rho^2]\right)}{d^2-1}\alpha^2$ & $d-\tr[\rho^2]$ & $d-\tr[\rho^2]$ \\
\hline\hline 
\end{tabular*}
\end{table}

Similar to \cref{th:corre_c_sep_cri}, \cref{th:univ_sep_cri_gamma} leads to the same criteria in the same measurement orbit. In the \cref{tab:various_measurement_vs}, we list the values $\mathcal{S}$ and $\mathcal{V}(\rho,\boldsymbol{X})$ for the aforementioned measurements, which gives new separability criteria via \cref{th:univ_sep_cri_gamma}. And then, we compare the relations between these criteria based on these three measurements. Since the separability is invariant under the SLOCC transformation \cite{verstraete03,li18-frAQB}, we consider the normal form of bipartite state \cref{eq:bipartite_state}
\begin{align}
\tilde{\rho} = \frac{1}{d^2}\mathds{1}\otimes\mathds{1} + \sum_{\mu\nu}\widetilde{\chi}_{\mu\nu}\pi_{\mu}\otimes\pi_{\nu} \; .
\label{eq:normal_form}
\end{align}
The following corollaries give the relations between criteria based on \cref{th:univ_sep_cri} and \cref{th:univ_sep_cri_gamma} respectively (see \cref{appen:coro_proof} for proof).
\begin{corollary}\label{coro:sep_norm}
Under the normal form, the \cref{th:univ_sep_cri} gives the equivalent separability criteria for measurements SCM, $\{\frac{h}{\sqrt{d}}\mathds{1},\frac{\pi_{\mu}}{\sqrt{2}}\}$ and $\{X_{\mu}|\tr[X_{\mu}X_{\nu}]=\delta_{\mu\nu}\}$
\begin{align}
&\|\widetilde{\chi}\|_{\mathrm{tr}} \leq \frac{2(d-1)}{d} \; .
\end{align}
\end{corollary}
\begin{corollary}\label{coro:gene_sep}
\cref{th:univ_sep_cri_gamma} gives the equivalent separability criteria for measurements SCM, $\{\frac{h}{\sqrt{d}}\mathds{1},\frac{\pi_{\mu}}{\sqrt{2}}\}$ and $\{X_{\mu}|\tr[X_{\mu}X_{\nu}]=\delta_{\mu\nu}\}$
\begin{align}
\|\chi'-\chi\|_{\mathrm{tr}} \leq 2 - \tr[\rho_{A}^2] - \tr[\rho_{B}^2] \; .
\end{align}
\end{corollary}
In Ref. \cite{shang18}, authors conjecture that ESIC criterion \cite{shang18} is stronger than CCNR criterion \cite{rudolph03,chen03}. \cref{coro:sep_norm} shows that ESIC criterion is equivalent to CCNR criterion under the normal form. For illustrations, we calculate the entangled regions for Bell diagonal states $\rho=\frac{1}{4}(\mathds{1}\otimes\mathds{1}+\sum_{\mu}t_{\mu}\sigma_{\mu}\otimes\sigma_{\mu})$ \cite{horodecki96-1tBDK} by \cref{coro:sep_norm}, i.e. $|t_1|+|t_2|+|t_3|> 1$, which is a necessary and sufficient condition. Thus, considering that SLOCC transformation, \cref{coro:sep_norm} is a necessary and sufficient condition of entanglement for 2-qubit system. Next, we compare \cref{coro:gene_sep} with the known criteria by $3\times 3$ Horodecki's bound entangled states with the white noise
\begin{align}
\rho(t,p) = p\rho_{\mathrm{\sss H}}(t) + (1-p)\frac{\mathds{1}}{9} \; .
\end{align}
Here, $\rho_{\mathrm{\sss H}}(t)$ is Horodecki's bound entangled states \cite{horodecki97}. \cref{fig:noise_horodecki} shows detection results of the three criteria (CCNR, ESIC, and \cref{coro:gene_sep}). In this case, Corollary 5 detects the more entangled states.
\begin{figure}
\centering
\includegraphics[width=0.6\linewidth]{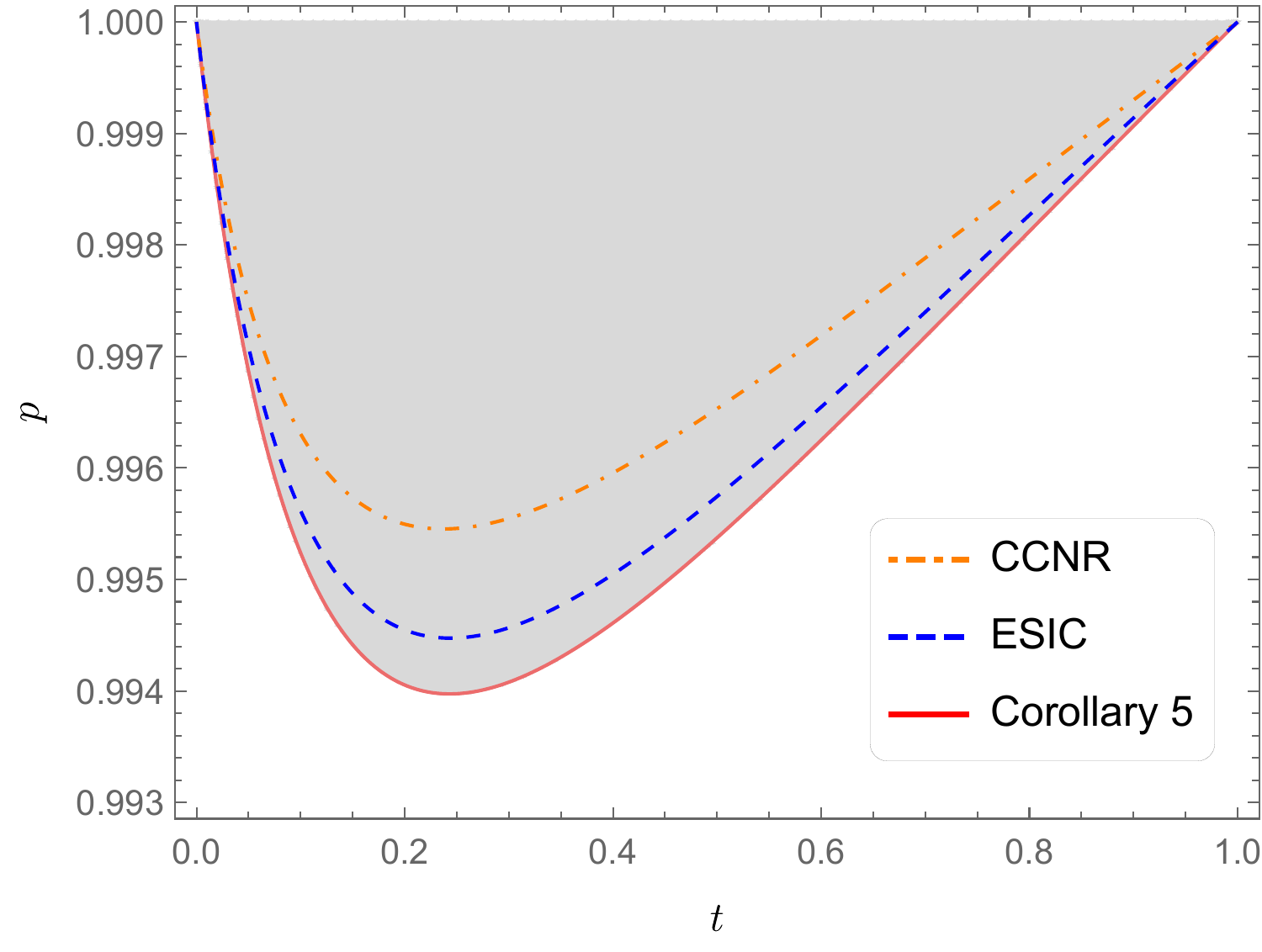}
\caption{Entanglement detection of Horodecki's $3\times 3$ bound entangled states with white noise $\rho(t,p)$. The orange dashdotted, blue dashed, and red solid curves are detection results of CCNR, ESIC, and \cref{coro:gene_sep} respectively. The light gray region denotes the entangled region. Corollary 5 detects the more entangled states.}
\label{fig:noise_horodecki}
\end{figure}

\subsection{Steering witness}
\noindent
In 2007, Wiseman \emph{et al.} \cite{wiseman07} formulated the concept of Schrödinger's steering \cite{schrodinger35,schrodinger36}, i.e. the quantum state $\rho$ will be steerable (by Alice), if the joint probability distribution $P(a,b|X^{A},X^{B};\rho)$ cannot be expressed as the following local hidden state model
\begin{align}
P(a,b|X^{A},X^{B};\rho) = \sum_{\xi}\wp(a|X^{A},\xi)\tr[\Pi_{b}^{B}\rho_{\xi}]\wp_{\xi} \; .
\label{eq:steering_def}
\end{align}
Here, $\wp(a|X^{A},\xi)$ is any possible probability distribution of Alice and $\Pi_{b}^{B}$ is the projector on the Bob's subspace associated to the measurement result $b$. $\wp_{\xi}$ is a normalized distribution involving hidden variable $\xi$. Equivalently, the separable states \cref{eq:sep_den_mat} also can be defined by the following similar joint probability distribution
\begin{align}
P(a,b|X^{A},X^{B};\rho) = \sum_{\xi}\tr[\Pi_{a}^{A}\rho_{\xi}^{A}]\tr[\Pi_{b}^{B}\rho_{\xi}^{B}]\wp_{\xi} \; .
\label{eq:sep_prob_def}
\end{align}
Here, $\Pi_{a}^{A}$ is similar to $\Pi_{b}^{B}$ and $\rho_{\xi}^{A}=\ket{\psi_{\xi}}\bra{\psi_{\xi}},\rho_{\xi}^{B}=\ket{\phi_{\xi}}\bra{\phi_{\xi}}$. Considering that the similarity between definitions of entanglement and steering, we can detect steering by similar methods. It is found that the entanglement criteria based on correlation matrices $\mathcal{C}$ and $\gamma$ can be generalized to detect quantum steering. We obtain the following steering criterion based on correlation $\mathcal{C}$ (see \cref{appen:univ_steering_cri} for proof).
\begin{theorem}\label{th:univ_steering_cri}
For arbitrary measurements $\boldsymbol{X}^{A}$ and $\boldsymbol{X}^{B}$, the correlation matrix $\mathcal{C}$ of the unsteerable states satisfies
\begin{align}
\|\mathcal{C}\|_{\mathrm{tr}} \leq \kappa \; ,
\end{align}
otherwise it is a steerable state (by Alice). Here, $\kappa=\kappa_{A}\kappa_{B}$, $\displaystyle\kappa_{A}=\min_{\boldsymbol{X}^{A}\in\mathcal{O}(\boldsymbol{X}^A)}\sqrt{\sum_{\mu}\lambda_{\mathrm{max}}((X_{\mu}^{A})^2)}$, $\displaystyle\kappa_{B}=\max_{\vec{x}^{B}\in\mathcal{B}(\boldsymbol{X}^{B})}|\vec{x}^{B}|$ and $\lambda_{\mathrm{max}}(X)$ refers to the maximal eigenvalue of $X$.
\end{theorem}
For illustration, we calculate 2-qubit Werner states $\rho_{\mathrm{W}}=p\ket{\psi_{s}}\bra{\psi_{s}}+\frac{1-p}{4}\mathds{1}$ and $\ket{\psi_{s}}$ is the singlet state. In case of two settings, i.e. $\boldsymbol{X}^{A}=\boldsymbol{X}^{B}=(\vec{n}_{1}\cdot\vec{\sigma},\vec{n}_{2}\cdot\vec{\sigma})^{\mathrm{T}}$ and the included angle $\braket{\vec{n}_{1},\vec{n}_{2}}=\delta$, we have $\kappa_{A}\leq \sqrt{2},\kappa_{B}=\sqrt{2+\sin2\delta}$ and $\|\mathcal{C}\|_{\mathrm{tr}}=2p$. Thus, Werner states are steerable (by Alice) if $p>\sqrt{1+\sin2\delta/2}$ and when $\delta=3\pi/4$, we can detect the most steerable states $p>1/\sqrt{2}$ in this scenario which agrees with the result in Ref. \cite{saunders10}. For orthogonal measurement $\boldsymbol{X}^{A}=\boldsymbol{X}^{B}=(\frac{\mathds{1}}{\sqrt{2}},\frac{\sigma_{x}}{\sqrt{2}},\frac{\sigma_{y}}{\sqrt{2}},\frac{\sigma_{z}}{\sqrt{2}})^{\mathrm{T}}$, we have $\kappa_{A}\leq \sqrt{2},\kappa_{B}=1$ and $\|\mathcal{C}\|_{\mathrm{tr}}=(3p+1)/2$, which shows that Werner states are steerable (by Alice) if $p>(2\sqrt{2}-1)/3$. In Ref. \cite{wiseman07,li21}, authors show that Werner states are steerable iff $p>1/2$ with infinite settings. It is possible to detect the more steerable states if we employ the more observables. Our method provides alternative simple method for steering witness.

And then, we give the steering criteria based on the correlation matrix $\gamma$ and LUR. Via \cref{eq:steering_def} and the concavity of the variance, it is readily to reach the following inequality of the joint observable $X^{A}+X^{B}$ for the unsteerable state $\rho_{\mathrm{ns}}$
\begin{align}
V(\rho_{\mathrm{ns}},X^{A}+X^{B}) \geq \sum_{\xi}\wp_{\xi}\left[V(\wp(a|X^{A},\xi))+V(\rho_{\xi},X^{B})\right] \; ,
\end{align}
where $V(\wp(a|X^{A},\xi))$ is the variance of distribution $\wp(a|X^{A},\xi)$. Immediately, for arbitrary measurements $\boldsymbol{X}^{A}$ and $\boldsymbol{X}^{B}$, we have
\begin{align}
\mathcal{V}(&\rho_{\mathrm{ns}},\boldsymbol{X}^{A}+\boldsymbol{X}^{B}) \geq \sum_{\xi}\wp_{\xi}\left[\sum_{\mu}V(\wp(a_{\mu}|X_{\mu}^{A},\xi))+\mathcal{V}(\rho_{\xi},\boldsymbol{X}^{B})\right] \; .
\end{align}
Here, $\{a_{\mu}\}$ are eigenvalues of $X_{\mu}^{A}$. Since Alice's site can design any possible probability distribution, the incompatibility of the measurement $\boldsymbol{X}^{A}$ can reach zero i.e., $\sum_{\mu}V(\wp(a|X_{\mu}^{A},\xi))=0$. Thus, we have
\begin{align}
\mathcal{V}(\rho_{\mathrm{ns}},\boldsymbol{X}^{A}+\boldsymbol{X}^{B}) \geq \mathcal{S}_{B} \; .
\label{eq:steering_lur_crit}
\end{align}
Thus, we recover the steering criterion based on LUR in Ref. \cite{zhen16}, which is very similar to the entanglement case \cref{eq:lur_crit}. Accordingly, we also obtain the following more effective LUR steering criterion.
\begin{theorem}
For arbitrary measurements $\boldsymbol{X}^{A}$ and $\boldsymbol{X}^{B}$, the unsteerable states satisfy the following local uncertainty relation
\begin{align}
\min_{\boldsymbol{X}^{A}\in\mathcal{O}(\boldsymbol{X}^A),\boldsymbol{X}^{B}\in\mathcal{O}(\boldsymbol{X}^B)}\mathcal{V}(\rho,\boldsymbol{X}^{A}+\boldsymbol{X}^{B}) \geq \mathcal{S}_{B} \; ,
\label{eq:steering_opt_lur_crit}
\end{align}
otherwise, it will be steerable (by Alice). Here, $\mathcal{S}_{B}$ is the optimal uncertainty bound of the Bob's subsystem, that is, $\mathcal{S}_{B}=\min_{\rho^{B}}\mathcal{V}(\rho^{B},\boldsymbol{X}^{B})$.
\label{th:opt_steering_lur_crit}
\end{theorem}
Similar to entanglement case, if we detect steering by \cref{eq:steering_lur_crit} straightforwardly, it is likely to fail. \cref{th:opt_steering_lur_crit} automatically optimizes measurements in the measurement orbits. In Ref. \cite{ji15,zhen16}, authors also derive steering criteria based on uncertainty relations and they show that how to optimize measurement of Alice's site which results in strong criteria. However, the optimization procedures rely on quantum state to be tested and when there is a lack of knowledge about quantum state it will be not effective enough. \cref{th:opt_steering_lur_crit} can automatically optimize measurements in the measurement orbits which ensures the effectiveness of this criterion in most cases even though quantum states are unknown. Equivalently, we can reformulate the criterion via correlation matrix $\gamma$, i.e.
\begin{align}
\|\gamma\|_{\mathrm{tr}} \leq \frac{\mathcal{V}(\rho_{A},\boldsymbol{X}^{A}) + \mathcal{V}(\rho_{B},\boldsymbol{X}^{B}) - \mathcal{S}_{B}}{2} \; .
\end{align}
It is possible to enhance the above steering criteria again if we allow to adjust Alice's measurement $\boldsymbol{X}^{A}$. To be specific, let Alice's observables contain real parameters $\xi_{\mu}$, i.e. $\xi_{\mu}X_{\mu}^{A}$ and then we have
\begin{align}
\gamma'_{\mu\nu} &= \braket{\xi_{\mu}X_{\mu}^{A}}\braket{X_{\nu}^{B}} - \braket{\xi_{\mu}X_{\mu}^{A}\otimes X_{\nu}^{B}} \\
&= \xi_{\mu}\gamma_{\mu\nu} \; ,
\end{align}
whose matrix form is $\gamma'=\Xi\circ\gamma$. Here, $\Xi=(\vec{\xi},\cdots,\vec{\xi})$ is a $m\times m$ matrix and $\vec{\xi}=(\xi_{1},\cdots,\xi_{m})^{\mathrm{T}}$. $A\circ B$ denotes Hadamard product, i.e $(A\circ B)_{\mu\nu}=A_{\mu\nu}B_{\mu\nu}$. The enhanced steering criterion follows.
\begin{theorem}
For arbitrary measurements $\boldsymbol{X}^{A}$ and $\boldsymbol{X}^{B}$, the correlation matrix $\gamma$ of the unsteerable states satisfies
\begin{align}
\|\Xi\circ\gamma\|_{\mathrm{tr}} \leq \frac{\mathcal{V}(\rho_{A},\vec{\xi}\circ\boldsymbol{X}^{A}) + \mathcal{V}(\rho_{B},\boldsymbol{X}^{B}) - \mathcal{S}_{B}}{2} \; ,
\end{align}
otherwise it is a steerable state (by Alice) and $\vec{\xi}\circ\boldsymbol{X}^{A}=(\xi_{1}X_{1}^{A},\cdots,\xi_{m}X_{m}^{A})^{\mathrm{T}}$.
\label{th:univ_steering_cri_gamma}
\end{theorem}
If only consider the single parameter, i.e. $\xi_{1}=\cdots=\xi_{m}=\xi$, we have the following simple corollary.
\begin{corollary}
For arbitrary measurements $\boldsymbol{X}^{A}$ and $\boldsymbol{X}^{B}$, the correlation matrix $\gamma$ of the unsteerable states satisfies
\begin{align}
\|\gamma\|_{\mathrm{tr}} \leq \frac{\xi^2\mathcal{V}(\rho_{A},\boldsymbol{X}^{A}) + \mathcal{V}(\rho_{B},\boldsymbol{X}^{B}) - \mathcal{S}_{B}}{2\xi} \; ,
\end{align}
otherwise it is a steerable state (by Alice) and $\xi$ is a positive real parameter.
\label{coro:univ_steering_cri_gamma}
\end{corollary}
Considering that quantum states with normal form \cref{eq:normal_form}, they have the maximally disordered subsystems, i.e. $\rho_{A}=\rho_{B}=\frac{1}{d}\mathds{1}$, which implies $\mathcal{V}(\rho_{A},\vec{\xi}\circ\boldsymbol{X}^{A})=2\sum_{\mu}\xi_{\mu}^{2}(\hat{n}_{\mu}^{A})^{2}/d$, $\mathcal{V}(\rho_{B},\boldsymbol{X}^{B})=2\sum_{\mu}(\hat{n}_{\mu}^{B})^{2}/d$ and $X_{\mu}^{A}=\vec{n}_{\mu}^{A}\cdot \boldsymbol{\Pi}=n_{0\mu}^{A}\Pi_{0} + \hat{n}_{\mu}^{A}\cdot\vec{\pi}$, $X_{\mu}^{B}=\vec{n}_{\mu}^{B}\cdot \boldsymbol{\Pi}=n_{0\mu}^{B}\Pi_{0} + \hat{n}_{\mu}^{B}\cdot\vec{\pi}$. Therefore, there is corollary as follows.
\begin{corollary}
For arbitrary measurements $\boldsymbol{X}^{A}$ and $\boldsymbol{X}^{B}$, the correlation matrix $\gamma$ of the unsteerable states with normal form \cref{eq:normal_form} satisfies
\begin{align}
\|\Xi\circ\gamma\|_{\mathrm{tr}} \leq \frac{\displaystyle 2\sum_{\mu=1}^{m}\xi_{\mu}^{2}(\hat{n}_{\mu}^{A})^{2} + 2\sum_{\mu=1}^{m}(\hat{n}_{\mu}^{B})^{2} - d\mathcal{S}_{B}}{2d} \; ,
\end{align}
otherwise it is a steerable state (by Alice).
\label{coro:slocc_steering_cri_gamma}
\end{corollary}
And then, we test the strength of our steering criteria via some specific examples.

\emph{Random two-qubit states with the Hilbert-Schmidt distribution.} If a square random matrix $A$ whose entries are drawn independently according to the standard complex Gaussian distribution, then the density matrix
$\rho = AA^{\dagger}/\tr[AA^{\dagger}]$ obeys the Hilbert-Schmidt distribution \cite{shang18}. To test \cref{coro:univ_steering_cri_gamma}, we choose measurement $\boldsymbol{X}^{A}=\boldsymbol{X}^{B}=\vec{\sigma}$. \cref{fig:random_qubits_cg} illustrates the results on 50 000 two-qubit states which are generated randomly according to the Hilbert-Schmidt distribution. As can be seen, the strength of criterion depends on parameter $\xi$ and the most about $4\%$ steerable states can be detected with parameter about $0.54$.
\begin{figure}
\centering
\includegraphics[width=0.6\linewidth]{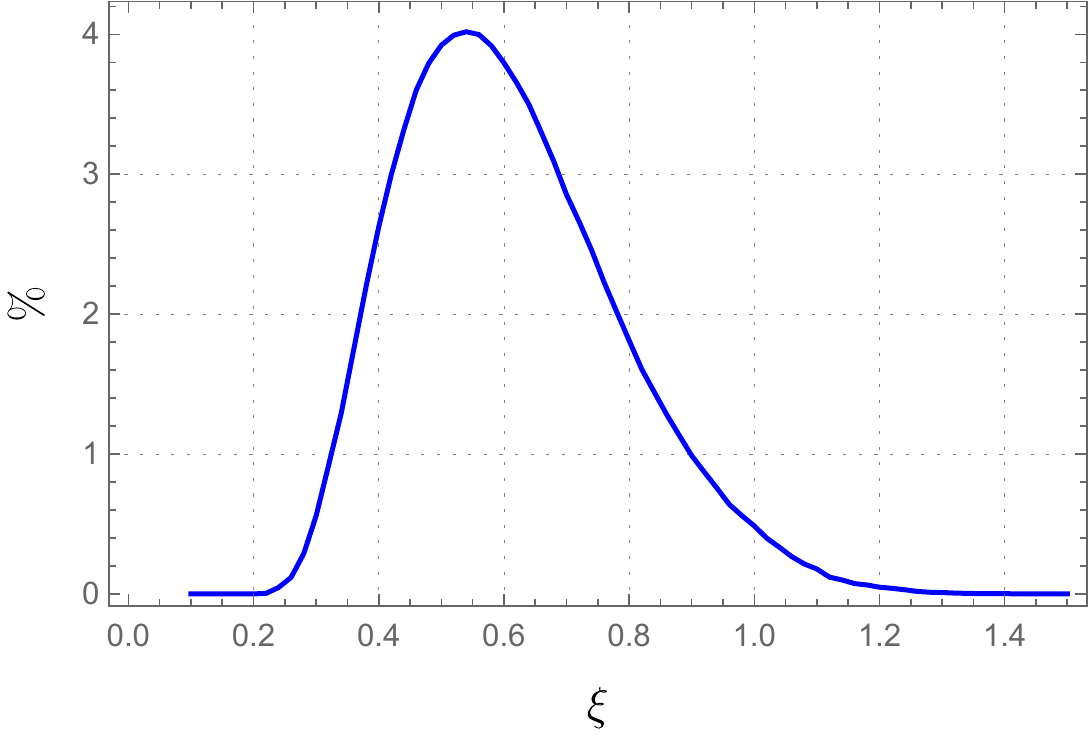}
\caption{Randomly generated two-qubit states. The plot illustrates the strength of \cref{coro:univ_steering_cri_gamma} on 50 000 two-qubit states which are generated randomly according to the Hilbert-Schmidt distribution. The strength of criterion depends on parameter $\xi$ and the most about $4\%$ steerable states can be detected with parameter about $0.54$.}
\label{fig:random_qubits_cg}
\end{figure}

\emph{Bell diagonal states.} We consider the following Bell diagonal states \cite{horodecki96-1tBDK}
\begin{align}
\rho_{\mathrm{BD}} = \frac{1}{4}\left(\mathds{1}\otimes \mathds{1} + \sum_{\mu=1}^{3}t_{\mu}\sigma_{\mu}\otimes \sigma_{\mu}\right) \; .
\label{eq:belldiag}
\end{align}
We still choose Pauli matrix $\boldsymbol{X}^{A}=\boldsymbol{X}^{B}=\vec{\sigma}$ and let $\xi_{\mu}=t_{\mu}$, then we have $\|\Xi\circ\gamma\|_{\mathrm{tr}}=\sum_{\mu=1}^{3}t_{\mu}^{2}$. Via \cref{coro:slocc_steering_cri_gamma}, if $\rho_{\mathrm{BD}}$ is steerable (by Alice), then we have
$t_{1}^2+t_{2}^2+t_{3}^2> 1$.

\emph{Isotropic states.} Via the Bloch representation, isotropic states can be reformulated as \cite{yang21}
\begin{align}
\rho_{\mathrm{I}} = \frac{1}{d^2} \mathds{1}\otimes \mathds{1} + \frac{1}{4} \sum_{\mu=1}^{d^2-1} \frac{2\eta}{d} \pi_{\mu} \otimes \pi_{\mu}^{\mathrm{T}} \; .
\label{eq:bloch_isotropic}
\end{align}
Here, $0<\eta\leq 1$. Choosing measurement $\boldsymbol{X}^{A}=\boldsymbol{X}^{B}=\vec{\pi}$, then we have $\mathcal{V}(\rho,\vec{\pi})=2(d-\tr[\rho^2])$. Let $\xi=\eta$, then $\rho_{\mathrm{I}}$ is steerable (by Alice), if $\|\gamma\|_{\mathrm{tr}}=2\eta(d^2-1)/d>\left(\eta^2(d^2-1)+(d-1)\right)/d\eta$, i.e. $\eta >1/\sqrt{d+1}$ which agrees with the steering criteria from entropic uncertainty relation \cite{costa18}.

\section{Conclusions}
\noindent
In this paper, we consider the entanglement and steering detection in the perspective of measurement. Criteria for Entanglement are established based on the correlation matrices $\mathcal{C}$ and $\gamma$ of arbitrary measurements, which in some sense provides a unified framework for the study of entanglement. The concept of measurement orbit is proposed, which enables the understanding of entanglement detection more transparent. As a simple corollary of the correlation matrix $\mathcal{C}$ criterion, we find ESIC and CCNR criteria are equivalent in normal form. The criteria via correlation matrix $\mathcal{C}$ are found equivalent to the entanglement witness, which paves a way to construct entanglement witness for arbitrary measurement. Variance uncertainty relation is obtained by means of measurement orbit, applicable to the development of effective entanglement criteria. For qudit system, we introduce the concept of SCM, a generalization of SIC-POVM, employed to reconstruct quantum state analytically. With measurement orbit, the equivalence of the entanglement criteria from matrix $\gamma$ and LUR is proved. Naturally, these methods can be transplanted to the steering detection. Steering criteria based on the correlation matrices $\mathcal{C}$ and $\gamma$ are obtained, and the equivalent LUR criteria are established. To detect steering effectively, performing proper measurement is crucial. The criteria in this work can automatically optimize the measurement, even without the knowledge of quantum state.

\section*{Acknowledgements}
\noindent
This work was supported in part by the National Natural Science Foundation of China(NSFC) under the Grants 11975236 and by the University of Chinese Academy of Sciences.

\section*{Author Contributions}
\noindent
All authors have equally contributed to the main result, the examples and the writing. All authors have given approval for the final version of the manuscript.

\section*{Competing Interests}
\noindent
The authors declare no competing interests.

\section*{Data Availability}
\noindent
All relevant data used for Examples and Figs. are available from the authors.

\appendix
\section{The proof of \cref{prop:state_rep_property}}\label{appen:state_rep_property}
\begin{proof}
---The convexity of the set $\mathcal{B}(\boldsymbol{X})$ is related to the convexity of quantum state $\rho$. We know that the set of quantum states is the convex hull of pure states $\ket{\psi_{k}}\bra{\psi_{k}}$, i.e.
\begin{align}
\rho = \sum_{k}p_{k}\ket{\psi_{k}}\bra{\psi_{k}} \; .
\end{align}
Here, $\sum_{k}p_{k}=1$ and $p_{k}>0$. So, we have $x_{\mu}=\sum_{k}p_{k}\braket{\psi_{k}|X_{\mu}|\psi_{k}}$ and $|\vec{x}|=\sqrt{\sum_{\mu}\tr[\rho X_{\mu}]^2}<\infty$, which demonstrates that $\mathcal{B}(\boldsymbol{X})$ is bounded by a convex hull constructed by vectors $\vec{x}$ of corresponding to pure states. Assuming that $X'_{\mu}=\sum_{\nu}O_{\mu\nu}X_{\nu}$, we have
\begin{align}
x'_{\mu} &= \tr[\rho X'_{\mu}] = \sum_{\nu}O_{\mu\nu}\tr[\rho X_{\nu}] \\
&= \sum_{\nu}O_{\mu\nu}x_{\nu} \; ,
\end{align}
whose matrix form is $\vec{x}'=O\,\vec{x}$. Thus, $\mathcal{B}(\boldsymbol{X}')=\{\vec{x}'=O\,\vec{x}|x_{\mu}=\tr[\rho X_{\mu}],\forall \rho\}$.
\end{proof}

\section{The proof of \cref{eq:purity_equality}}\label{appen:purity_equality}
\noindent
The Eq. (15) in the main text is
\begin{align}
\tr[\rho^2] = \vec{\omega}^{\mathrm{T}}\Omega\vec{\omega} = \vec{x}^{\mathrm{T}}\Omega^{-}\vec{x} \; .
\end{align}
\begin{proof}
---A density matrix can be expanded by $X_{\mu}$ linearly
\begin{align}
\rho = \sum_{\mu}\omega_{\mu}X_{\mu} \; .
\end{align}
Thus
\begin{align}
\tr[\rho^2] &= \sum_{\mu\nu}\omega_{\mu}\tr[X_{\mu}X_{\nu}]\omega_{\nu} \\
&= \sum_{\mu\nu}\omega_{\mu}\Omega_{\mu\nu}\omega_{\nu} \\
&= \vec{\omega}^{\mathrm{T}}\Omega\vec{\omega} \; .
\end{align}
Here, $\Omega_{\mu\nu}=\tr[X_{\mu}X_{\nu}]$. $\vec{\omega}$ can be viewed as the general solution of linear equation $\Omega \vec{\omega}=\vec{x}$ which can be obtained via the Moore-Penrose inverse $\Omega^{-}$ \cite{penrose55,ben-israel03}
\begin{align}
\vec{\omega} = \Omega^{-}\vec{x} + (\mathds{1}-\Omega^{-}\Omega)\vec{x}_{0} \; .
\end{align}
Here, $\Omega^{-}$ is referred as the Moore-Penrose inverse of $\Omega$ and satisfies the following Moore-Penrose conditons \cite{penrose55}
\begin{align}
\label{eq:moore_penrose_cond1}
&\Omega\Omega^{-}\Omega = \Omega \; , \\
\label{eq:moore_penrose_cond2}
&\Omega^{-}\Omega\Omega^{-} = \Omega^{-} \; , \\
\label{eq:moore_penrose_cond3}
&(\Omega\Omega^{-})^{\dagger} = \Omega\Omega^{-} \; , \\
\label{eq:moore_penrose_cond4}
&(\Omega^{-}\Omega)^{\dagger} = \Omega^{-}\Omega \; .
\end{align}
It is noted that the Moore-Penrose inverse $\Omega^{-}$ satisfying the four conditions is unique. In light of these conditions of the Moore-Penrose inverse $\Omega^{-}$, we have
\begin{align}
\tr[\rho^2] &= \vec{\omega}^{\mathrm{T}}\Omega\vec{\omega} \notag \\
&= \left(\Omega^{-}\vec{x} + (\mathds{1}-\Omega^{-}\Omega)\vec{x}_{0}\right)^{\mathrm{T}}\Omega\left(\Omega^{-}\vec{x} + (\mathds{1}-\Omega^{-}\Omega)\vec{x}_{0}\right) \notag \\
&= \left(\vec{x}^{\mathrm{T}}\Omega^{-} + \vec{x}_{0}^{\mathrm{T}}(\mathds{1}-\Omega\Omega^{-})\right)\Omega\left(\Omega^{-}\vec{x} + (\mathds{1}-\Omega^{-}\Omega)\vec{x}_{0}\right) \notag \\
&= \vec{x}^{\mathrm{T}}\Omega^{-}\Omega\Omega^{-}\vec{x} + \vec{x}^{\mathrm{T}}\Omega^{-}\Omega(\mathds{1}-\Omega^{-}\Omega)\vec{x}_{0} +  \notag \\
&\vec{x}_{0}^{\mathrm{T}}(\mathds{1}-\Omega\Omega^{-})\Omega\Omega^{-}\vec{x} + \vec{x}_{0}^{\mathrm{T}}(\mathds{1}-\Omega\Omega^{-})\Omega(\mathds{1}-\Omega^{-}\Omega)\vec{x}_{0} \notag \\
&= \vec{x}^{\mathrm{T}}\Omega^{-}\vec{x}
\end{align}
Here, we also make use of $\Omega=\Omega^{\mathrm{T}}$ and $(\Omega^{-})^{\mathrm{T}}=(\Omega^{\mathrm{T}})^{-}$.
\end{proof}

\section{The proof of \cref{prop:variance_sum_prop}}\label{appen:variance_sum_prop}
\begin{proof}
---The first conclusion can be straightforwardly proved by definition of variance. For the second one, considering that two measurements $\boldsymbol{X},\boldsymbol{X}'$ and $X'_{\mu}=\sum_{\nu}O_{\mu\nu}X_{\nu}$, we have
\begin{align}
\sum_{\mu}\tr[\rho X_{\mu}'^2] &= \sum_{\mu}\tr[\rho\sum_{\nu\nu'}O_{\mu\nu}O_{\mu\nu'}X_{\nu}X_{\nu'}] \\
&= \sum_{\nu\nu'}\tr[\rho\sum_{\mu}O_{\mu\nu}O_{\mu\nu'}X_{\nu}X_{\nu'}] \\
&= \sum_{\nu\nu'}\tr[\rho\delta_{\nu\nu'}X_{\nu}X_{\nu'}] \\
&= \sum_{\nu}\tr[\rho X_{\nu}^2] \; .
\end{align}
Here, $OO^{\mathrm{T}}=\mathds{1}$. Similarly, we have $\sum_{\mu}\tr[\rho X_{\mu}]^2=\sum_{\mu}\tr[\rho X'_{\mu}]^2$. Thus,
\begin{align}
\mathcal{V}(\rho,\boldsymbol{X}) = \mathcal{V}(\rho,\boldsymbol{X}') \; .
\end{align}
The third one is due to the additivity of variance i.e.
\begin{align}
V&\left(\bigotimes_{k}\rho^{(k)},\sum_{i}X^{(i)}\right) \notag \\
=& \tr\left[\bigotimes_{k}\rho^{(k)}\left(\sum_{i}X^{(i)}\right)^2\right] - \tr\left[\bigotimes_{k}\rho^{(k)}\sum_{i}X^{(i)}\right]^2 \notag \\
=& \sum_{ij}\tr\left[\bigotimes_{k}\rho^{(k)}X^{(i)}X^{(j)}\right] - \notag \\ &\sum_{ij}\tr\left[\bigotimes_{k}\rho^{(k)}X^{(i)}\right]\tr\left[\bigotimes_{k}\rho^{(k)}X^{(j)}\right] \notag \\
=& \sum_{i}\tr\left[\rho^{(i)}(X^{(i)})^2\right] + \sum_{i\neq j}\tr\left[\bigotimes_{k}\rho^{(k)}X^{(i)}X^{(j)}\right] \notag \\
& - \sum_{i}\tr\left[\rho^{(i)}X^{(i)}\right]^2 - \sum_{i\neq j}\tr\left[\bigotimes_{k}\rho^{(k)}X^{(i)}X^{(j)}\right] \notag \\
=& \sum_{i}V(\rho^{(i)},X^{(i)}) \; .
\end{align}
The fourth one is due to the concavity of variance i.e.
\begin{align}
V\left(\sum_{i}p_i\rho_i,X\right) &= \sum_{i}p_{i}\mathrm{Tr}[\rho_{i} X^2] - \left(\sum_{i}p_{i}\mathrm{Tr}[\rho_{i} X]\right)^2 \notag \\
&\geq \sum_{i}p_{i}\mathrm{Tr}[\rho_{i} X^2] - \sum_{i}p_{i}\mathrm{Tr}[\rho_{i} X]^2 \notag \\
&= \sum_{i}p_{i}V\left(\rho_{i},X\right)
\end{align}
The inequality is due to the nonnegativity of the variance, that is, $\sum_{i}p_{i}\alpha_{i}^2-(\sum_{i}p_{i}\alpha_{i})^2\geq 0$, and $\alpha_{i}=\mathrm{Tr}[\rho_{i} X]$.
\end{proof}

\section{The proof of \cref{prop:scm_prope}}\label{appen:scm_prope}
\noindent
Firstly, we give an interesting lemma which is used in the next proof.
\begin{lemma}
Provided that we have a matrix in such a form
\begin{align}
A=\left(
\begin{matrix}
\vec{a}_{1} & \vec{a}_{2} & \cdots & \vec{a}_{d+1} \\
1/\sqrt{d+1} & 1/\sqrt{d+1} & \cdots & 1/\sqrt{d+1} \\
\end{matrix}
\right) \; .
\end{align}
Here, $\vec{a}_{i}$ are $d$-dimensional real column vectors and $|\vec{a}_{i}|=\sqrt{\frac{d}{d+1}}$. Then, $A$ is invertible if and only if $\{\vec{a}_{1}, \vec{a}_{2}, \cdots, \vec{a}_{d+1}\}$ configures a $d$-simplex and $A$ is orthogonal if and only if $\{\vec{a}_{1}, \vec{a}_{2}, \cdots, \vec{a}_{d+1}\}$ configures a regular $d$-simplex.
\label{lem:simplex_orth}
\end{lemma}
\begin{proof}
---If $A$ is invertible, the homogeneous linear equation group $A\vec{x} = 0$ has only null solution, that is
\begin{align}
\sum_{i=1}^{d+1}x_{i}\vec{a}_{i} = 0 \; , \; \sum_{i=1}^{d+1}x_{i} = 0 \; ,
\end{align}
if and only if $x_{1} = x_{2} = \cdots = x_{d+1} = 0$, which is equivalent to $\vec{a}_{1}, \vec{a}_{2}, \cdots, \vec{a}_{d+1}$ are affine independent and they constitute a $d$-simplex \cite{boyd06}. And if $\{\vec{a}_{1}, \vec{a}_{2}, \cdots , \vec{a}_{d+1}\}$ configures a regular $d$-simplex,
the angles $\theta$ between any two vectors of $\{\vec{a}_{1}, \vec{a}_{2}, \cdots, \vec{a}_{d+1}\}$ satisfy $\cos\theta = \frac{-1}{d}$. Then, it is readily to check $A^{\mathrm{T}}A=\mathds{1}$ as follow
\begin{align}
[A^{\mathrm{T}}A]_{ij} &= \vec{a}_{i}\cdot\vec{a}_{j} + \frac{1}{d+1} = \delta_{ij} \; .
\end{align}
Here, we have used $|\vec{a}_{i}|^2+\frac{1}{d+1}=1$ and $\vec{a}_{i}\cdot\vec{a}_{j}+\frac{1}{d}=\frac{d}{d+1}\cos\theta+\frac{1}{d+1}=0, i\neq j$. On the contrary, we can prove that $\{\vec{a}_{1}, \vec{a}_{2}, \cdots, \vec{a}_{d+1}\}$ configures a regular $d$-simplex via the columns orthonormality of $A$.
\end{proof}
The proof of Proposition 4.
\begin{proof}
---In light of representation introduced in the main text, an element of SCM can be parameterized as
\begin{align}
X_{\mu} = \vec{n}_{\mu}\cdot \Pi = \frac{h}{d}\mathds{1} + \hat{n}_{\mu}\cdot \vec{\pi} \; .
\end{align}
Considering that $\{\hat{n}_{\mu}\}_{\mu=1}^{d^2}$ constitutes a regular $d$-simplex, which implies $\sum_{\mu=1}^{d^2}\hat{n}_{\mu}=0$, we reach
\begin{align}
\sum_{\mu=1}^{d^2}X_{\mu} = dh\mathds{1} \; .
\end{align}
For an arbitrary measurements $\boldsymbol{X}$ with $\tr[X_{\mu}]=h$
\begin{align}
\sum_{\mu=1}^{m}X_{\mu}^2 &= \sum_{\mu=1}^{m}\left(\vec{n}_{\mu}\cdot\boldsymbol{\Pi}\right)^2 \\
&= \sum_{\mu=1}^{m}(n_{0}\Pi_{0} + \hat{n}_{\mu}\cdot\vec{\pi})^2 \\
&= mn_{0}^2\Pi_{0}^2 + 2n_{0}\Pi_{0}\sum_{\mu}\hat{n}_{\mu}\cdot\vec{\pi} + \sum_{\mu}(\hat{n}_{\mu}\cdot\vec{\pi})^2 \\
&= \frac{mh^2}{d^2}\mathds{1} + \frac{2h}{d}\sum_{\mu}\hat{n}_{\mu}\cdot\vec{\pi} + \mathcal{C}' \; .
\end{align}
Here, $\mathcal{C}'=\sum_{\mu}(\hat{n}_{\mu}\cdot\vec{\pi})^2=\sum_{k,l}N_{kl}\pi_{k}\pi_{l}$ and $N=\sum_{\mu}\hat{n}_{\mu}\hat{n}_{\mu}^{\mathrm{T}}$. The first term is constant matrix and the last two terms depend on structure of $\{\hat{n}_{\mu}\}$. The following matrix is an orthogonal matrix via \cref{lem:simplex_orth}
\begin{align}
A=\left(
\begin{matrix}
\vec{a}_{1} & \vec{a}_{2} & \cdots & \vec{a}_{d^2} \\
1/d & 1/d & \cdots & 1/d \\
\end{matrix}
\right) \; .
\end{align}
Here, $\vec{a}_{i}=\frac{\sqrt{d^2-1}}{d}\vec{e}_{i}, |\vec{e}_{i}|=1$ and $\{\vec{e}_{i}\}$ constitutes a regular $(d^2-1)$-simplex.

Thus
\begin{align}
AA^{\mathrm{T}} &= \left(
\begin{matrix}
\sum_{\mu}\vec{a}_{\mu}\vec{a}_{\mu}^{\mathrm{T}} & 0 \\
0 & 1 \\
\end{matrix}
\right) = \mathds{1} \; ,
\end{align}
and then
\begin{align}
\sum_{\mu}\vec{e}_{\mu}\vec{e}_{\mu}^{\mathrm{T}} = \frac{d^2}{d^2-1} \mathds{1} \; .
\end{align}
Thoughout the whole proof, $\mathds{1}$ denotes identity matrix in corresponding Hilbert space.
If $\boldsymbol{X}$ is an SCM, that is, $\{\hat{n}_{\mu}\}$ constitutes a regular $(d^2-1)$-simplex, we have
\begin{align}
N &=\sum_{\mu}\hat{n}_{\mu}\hat{n}_{\mu}^{\mathrm{T}} = \alpha^2\sum_{\mu}\vec{e}_{\mu}\vec{e}_{\mu}^{\mathrm{T}} \notag \\
&= \frac{d^2\alpha^2}{d^2-1} \mathds{1} \; .
\label{eq:sup_simp_scm}
\end{align}
Thus
\begin{align}
\mathcal{C}' &= \sum_{k,l}N_{kl}\pi_{k}\pi_{l} = \frac{d^2\alpha^2}{d^2-1}\sum_{k}\pi_{k}^2 \\
&= 2d\alpha^2\mathds{1} \; ,
\end{align}
where $\sum_{k}\pi_{k}^2=\frac{2(d^2-1)}{d}\mathds{1}$ is the quadratic Casimir operator of $\mathfrak{su}(d)$ Lie algebra. Then we reach
\begin{align}
\sum_{\mu=1}^{d^2}X_{\mu}^2 = (h^2 + 2d\alpha^2)\mathds{1} \; .
\end{align}
The third relation can be calculated straightforwardly
\begin{align}
\tr[X_{\mu}X_{\nu}] &= 2\vec{n}_{\mu}\cdot\vec{n}_{\nu} = 2n_{0}^2 + 2\hat{n}_{\mu}\cdot\hat{n}_{\nu} \\
&= \frac{h^2}{d} + \frac{2\alpha^2(d^2\delta_{\mu\nu}-1)}{d^2-1} \; .
\end{align}
Here, we have used $\cos\braket{\hat{n}_{\mu},\hat{n}_{\nu}}=\frac{-1}{d^2-1},\mu\neq\nu$. If $\alpha^2=\frac{d-1}{2d^3},h=\frac{1}{d}$ and $X_{\mu}\geq 0$,
\begin{align}
\sum_{\mu=1}^{d^2}X_{\mu} = \mathds{1} \; , \; \sum_{\mu=1}^{d^2}X_{\mu}^2 = \frac{1}{d}\mathds{1} \; , \; \tr[X_{\mu}X_{\nu}] = \frac{d\delta_{\mu\nu}+1}{d(d+1)} \; ,
\end{align}
which just is an SIC-POVM.
\end{proof}

\section{The proof of \cref{th:state_recon_scm}}\label{appen:state_recon_scm}
\noindent
Firstly, we calculate the Moore-Penrose inverse for the following form of matrix
\begin{align}
A = \xi \mathds{1} + \eta J_{d} \; ,
\end{align}
where, $J_{d}$ is a $d$-dimentional all-ones matrix all whose elements are equal to one and $\xi,\eta$ are two real parameters. Assuming that $A$ is invertible, its inverse can be expressed as
\begin{align}
A^{-1} = \xi' \mathds{1} + \eta' J_{d} \; .
\end{align}
Taking into matrix equation $AA^{-1}=\mathds{1}$, we have
\begin{align}
\xi' = \frac{1}{\xi} \; , \; \eta' = -\frac{\eta}{\xi(\xi+d\eta)} \; .
\end{align}
Thus, if $\xi\neq 0, -d\eta$, $A$ is invertible and $A^{-1}=\frac{1}{\xi} \mathds{1} - \frac{\eta}{\xi(\xi+d\eta)} J_{d}$. Here, we don't care about $\xi=0$. If $\xi=-d\eta$, $A=\eta(J_{d}-d\mathds{1})$ is singular but we can calculate its Moore-Penrose inverse
\begin{align}
A^{-} = \frac{1}{\eta}(\frac{1}{d^2}J_{d}-\frac{1}{d}\mathds{1}) \; .
\label{eq:moore_penrose_inv}
\end{align}
Here, if $\eta=0$, $A^{-}=0$. It can be checked that \cref{eq:moore_penrose_inv} satisfies the four Moore-Penrose conditions \cref{eq:moore_penrose_cond1,eq:moore_penrose_cond2,eq:moore_penrose_cond3,eq:moore_penrose_cond4}. The proof of Theorem 2 is as follow.
\begin{proof}
---As a complete measurement set, a density matrix can be reconstructed by an SCM
\begin{align}
\rho = \sum_{\mu}\omega_{\mu}X_{\mu} \; .
\end{align}
Here, $\sum_{\mu}\omega_{\mu}=\frac{1}{h}$ due to $\tr[\rho]=1$. For an SCM, we have
\begin{align}
\Omega_{\mu\nu} = \tr[X_{\mu}X_{\nu}] = \frac{h^2}{d} + \frac{2\alpha^2(d^2\delta_{\mu\nu}-1)}{d^2-1} \; ,
\end{align}
that is
\begin{align}
\Omega = \frac{2\alpha^2d^2}{d^2-1}\mathds{1} + (\frac{h^2}{d}-\frac{2\alpha^2}{d^2-1})J_{d^2} \; .
\end{align}
In light of above discussions, if $h\neq 0$, $\Omega$ is invertible and
\begin{align}
\Omega^{-1} = \frac{d^2-1}{2\alpha^2d^2}\mathds{1} + (\frac{1}{d^3h^2}-\frac{d^2-1}{2\alpha^2d^4})J_{d^2} \; .
\label{eq:scm_inve}
\end{align}
Thus
\begin{align}
\rho &= \sum_{\mu}\omega_{\mu}X_{\mu} = \sum_{\mu\nu}\Omega_{\mu\nu}^{-1}x_{\nu}X_{\mu} \\
&= \frac{d^2-1}{2\alpha^2d^2}\sum_{\mu}x_{\mu}X_{\mu} + (\frac{1}{d^3h^2}-\frac{d^2-1}{2\alpha^2d^4})\sum_{\nu}x_{\nu}\sum_{\mu}X_{\mu} \\
&= \frac{d^2-1}{2\alpha^2d^2}\sum_{\mu}x_{\mu}X_{\mu} + (\frac{1}{d^3h^2}-\frac{d^2-1}{2\alpha^2d^4})d^2h^2\mathds{1} \\
&= \frac{d^2-1}{2\alpha^2d^2}\sum_{\mu}x_{\mu}X_{\mu} + (\frac{1}{d}-\frac{(d^2-1)h^2}{2\alpha^2d^2})d^2h^2\mathds{1}
\end{align}
Here, we make use of $\sum_{\mu}X_{\mu}=dh\mathds{1}$ in the second line. In the main text, we reach the following equation
\begin{align}
\tr[\rho^2] = \vec{x}^{\mathrm{T}}\Omega^{-}\vec{x} \; .
\label{eq:tr_rho_square}
\end{align}
Taking \cref{eq:scm_inve} into \cref{eq:tr_rho_square}
\begin{align}
\tr[\rho^2] &= \frac{d^2-1}{2\alpha^2d^2}|\vec{x}|^2 + (\frac{1}{d^3h^2}-\frac{d^2-1}{2\alpha^2d^4})\vec{x}^{\mathrm{T}}J_{d^2}\vec{x} \\
&= \frac{d^2-1}{2\alpha^2d^2}|\vec{x}|^2 + (\frac{1}{d}-\frac{h^2(d^2-1)}{2\alpha^2d^2}) \; .
\end{align}
Here, $\vec{x}^{\mathrm{T}}J_{d^2}\vec{x}=(\sum_{\mu}x_{\mu})^2=d^2h^2$. Thus,
\begin{align}
|\vec{x}|^2 = \frac{2d\alpha^2\left(d\tr[\rho^2]-1\right)}{d^2-1} + h^2 \; .
\end{align}
If $\tr[X_{\mu}]=h=0$, in order to expand density matrix, we need identity matrix $\mathds{1}$ additional, i.e.
\begin{align}
\rho = \sum_{\mu=0}^{d^2}\omega_{\mu}X_{\mu} = \frac{1}{d}\mathds{1} + \sum_{\mu=1}^{d^2}\omega_{\mu}X_{\mu} \; .
\end{align}
Here, $X_{0}=\mathds{1}$ and $\omega_{0}=\frac{1}{d}$. Since $X_{\mu}$ are linearly dependent, the coefficients $\omega_{\mu}$ will be no unique. But we can still calculate $\tr[\rho^2]$ by \cref{eq:tr_rho_square}. When $h=0$, we have
\begin{align}
\Omega = d\oplus \Omega' \; ,
\end{align}
where, $\Omega'=-\frac{2\alpha^2}{d^2-1}( J_{d^2} - d^2\mathds{1})$. Via \cref{eq:moore_penrose_inv}
\begin{align}
\Omega'^{-} = \frac{d^2-1}{2\alpha^2d^2}(\mathds{1} - \frac{1}{d^2}J_{d^2}) \; .
\end{align}
Making use of $(A\oplus B)^{-}=A^{-}\oplus B^{-}$, we have
\begin{align}
\Omega^{-} = \frac{1}{d}\oplus \Omega'^{-} \; .
\end{align}
Thus
\begin{align}
\tr[\rho^2] &= \vec{x}^{\mathrm{T}}\Omega^{-}\vec{x} = \frac{1}{d}x_{o}^2 + \vec{x}^{\mathrm{T}}\Omega'^{-}\vec{x} \\
&=\frac{d^2-1}{2\alpha^2d^2}|\vec{x}|^2 + \frac{1}{d} \; .
\end{align}
Here, $x_{0}=\tr[\rho\mathds{1}]=1$. Thus,
\begin{align}
|\vec{x}|^2 = \frac{2d\alpha^2\left(d\tr[\rho^2]-1\right)}{d^2-1} \; .
\end{align}
\end{proof}

\section{The proof of \cref{eq:corre_matrix_c,eq:corre_matrix_gamma}}\label{appen:corre_matrix}
\noindent
The Eqs. (51), (52) in the main text is
\begin{align}
&\mathcal{C} = M_{A}^\mathrm{T}\chi M_{B} \; , \\
&\gamma = M_{A}^\mathrm{T}(\chi'-\chi)M_{B} \; .
\end{align}
Here, $M_{A}=(\vec{n}_{1}^{A},\vec{n}_{2}^{A},\cdots)$, $M_{B}=(\vec{n}_{1}^{B},\vec{n}_{2}^{B},\cdots)$ correspond to the implemented measurements in the two subsystems respectively and $\chi'_{\mu\nu}=\tr[(\rho_{A}\otimes\rho_{B})(\Pi_{\mu}\otimes\Pi_{\nu})]=\frac{d}{2}\chi_{\mu 0}\chi_{0\nu}$.
\begin{proof}
---The correlation matrices $\mathcal{C}$ and $\gamma$ are defined as
\begin{align}
&\mathcal{C}_{\mu\nu} = \braket{X_{\mu}^{A}\otimes X_{\nu}^{B}} \; , \\
&\gamma_{\mu\nu} = \braket{X_{\mu}^{A}}\braket{X_{\nu}^{B}} - \braket{X_{\mu}^{A}\otimes X_{\nu}^{B}} \; .
\end{align}
Here, $X_{\mu}^{A}=\vec{n}_{\mu}^{A}\cdot\boldsymbol{\Pi}$ and $X_{\nu}^{B}=\vec{n}_{\nu}^{B}\cdot\boldsymbol{\Pi}$. Considering that $\rho=\frac{1}{4}\sum_{ij}\chi_{ij}\Pi_{i}\otimes\Pi_{j}$
\begin{align}
\mathcal{C}_{\mu\nu} &= \tr\left[\rho(X_{\mu}^{A}\otimes X_{\nu}^{B})\right] \\
&= \tr\left[\frac{1}{4}\sum_{ij}\chi_{ij}\Pi_{i}\otimes\Pi_{j}\left(\sum_{kl}n_{\mu k}^{A}n_{\nu l}^{B}\Pi_{k}\otimes\Pi_{l}\right)\right] \\
&= \frac{1}{4}\sum_{ij}\chi_{ij}\sum_{kl}n_{\mu k}^{A}n_{\nu l}^{B}\tr\left[\Pi_{i}\Pi_{k}\otimes\Pi_{j}\Pi_{l}\right] \\
&= \sum_{ij}\chi_{ij}\sum_{kl}n_{\mu k}^{A}n_{\nu l}^{B}\delta_{ik}\delta_{jl} \\
&= \sum_{ij}n_{\mu i}^{A}\chi_{ij}n_{\nu j}^{B} \; ,
\end{align}
whose matrix form is
\begin{align}
\mathcal{C} = M_{A}^\mathrm{T}\chi M_{B} \; ,
\end{align}
and $M_{A}=(\vec{n}_{1}^{A},\vec{n}_{2}^{A},\cdots)$, $M_{B}=(\vec{n}_{1}^{B},\vec{n}_{2}^{B},\cdots)$. Here, we make use of the orthogonal relation $\tr[\Pi_{\mu}\Pi_{\nu}]=2\delta_{\mu\nu}$. For $\gamma$, we only need to notice
\begin{align}
\braket{X_{\mu}^{A}}\braket{X_{\nu}^{B}} &= \tr[(\rho_{A}\otimes\rho_{B})(X_{\mu}^{A}\otimes X_{\nu}^{B})] \\
&= \sum_{ij}n_{\mu i}^{A}\chi'_{ij}n_{\nu j}^{B}
\end{align}
Here, $\chi'_{ij}=\tr[(\rho_{A}\otimes\rho_{B})(\Pi_{i}\otimes \Pi_{j})]$. In addition, we have $\rho_{A}=\tr_{B}[\rho]=\frac{\sqrt{2d}}{4}\sum_{i}\chi_{i0}\Pi_{i}$ and $\rho_{B}=\tr_{A}[\rho]=\frac{\sqrt{2d}}{4}\sum_{j}\chi_{0j}\Pi_{j}$, thus $\chi'_{ij}=\frac{d}{2}\chi_{i0}\chi_{0j}$. Finally, we have
\begin{align}
\gamma = M_{A}^\mathrm{T}(\chi'-\chi)M_{B} \; .
\end{align}
\end{proof}

\section{The proof of \cref{ob:measu_orbit_invar}}\label{appen:measu_orbit_invar}
\begin{proof}
---We have proved that the correlation matrix $\mathcal{C}$ and $\gamma$ can be written as
\begin{align}
&\mathcal{C} = M_{A}^\mathrm{T}\chi M_{B} \; , \\
&\gamma = M_{A}^\mathrm{T}(\chi'-\chi)M_{B} \; .
\end{align}
Here, $M_{A}=(\vec{n}_{1}^{A},\vec{n}_{2}^{A},\cdots,\vec{n}_{m}^{A})$, $M_{B}=(\vec{n}_{1}^{B},\vec{n}_{2}^{B},\cdots,\vec{n}_{m}^{B})$ and $X_{\mu}^{A}=\vec{n}_{\mu}^{A}\cdot\boldsymbol{\Pi},X_{\nu}^{B}=\vec{n}_{\nu}^{B}\cdot\boldsymbol{\Pi}$. Assuming that there are two measurements $\boldsymbol{X}^{A}$ and $\boldsymbol{X}^{A'}$ lying in the same orbit, i.e. $X^{A'}_{\mu}=\sum_{\nu}O_{\mu\nu}^{A}X_{\nu}^{A}$ and $O^{A}(O^{A})^{\mathrm{T}}=\mathds{1}$. Via the orthogonal relation $\tr[\Pi_{\mu}\Pi_{\nu}]=2\delta_{\mu\nu}$, we have
\begin{align}
n_{\mu k}^{A'} = \sum_{\nu}O_{\mu\nu}^{A}n_{\nu k}^{A} \; ,
\end{align}
which can be rewritten as the following matrix form $M_{A}'^{\mathrm{T}} = O^{A}M_{A}^{\mathrm{T}}$. Similarly, we have $M_{B}'^{\mathrm{T}} = O^{B}M_{B}^{\mathrm{T}}$. Thus
\begin{align}
\|\mathcal{C}'\|_{\mathrm{tr}} &= \|M_{A}'^{\mathrm{T}}\chi M_{B}'\|_{\mathrm{tr}} \\
&= \|O^{A}(M_{A}^{\mathrm{T}}\chi M_{B})(O^{B})^{\mathrm{T}}\|_{\mathrm{tr}} \\
&= \|\mathcal{C}\|_{\mathrm{tr}} \; .
\end{align}
Similarly, we can prove that $\|\gamma\|_{\mathrm{tr}}$ also is a measurement orbit invariant.
\end{proof}

\section{The proof of \cref{th:corre_c_sep_cri}}\label{appen:corre_c_sep_cri}
\begin{proof}
---For a separable state $\rho = \sum_{k}p_{k}\ket{\psi_{k}}\bra{\psi_{k}}\otimes \ket{\phi_{k}}\bra{\phi_{k}}$, we have
\begin{align}
\mathcal{C}_{\mu\nu} &= \braket{X_{\mu}^{A}\otimes X_{\nu}^{B}} = \sum_{k}p_{k}\braket{X_{\mu}^{A}}_{k}\braket{X_{\nu}^{B}}_{k} \\
&= \sum_{k}p_{k}x_{\mu k}^{A}x_{\nu k}^{B} \; ,
\end{align}
which can be rewritten as the following matrix form
\begin{align}
\mathcal{C} = \sum_{k}p_{k}\vec{x}_{k}^{A}(\vec{x}_{k}^{B})^{\mathrm{T}} \; .
\end{align}
Here, $\vec{x}_{k}^{A}=(\braket{X_{1}^{A}}_{k},\braket{X_{2}^{A}}_{k},\cdots)^{\mathrm{T}}$ and similarly for $\vec{x}_{k}^{B}$. Thus
\begin{align}
\|\mathcal{C}\|_{\mathrm{tr}} &\leq \sum_{k}p_{k}\|\vec{x}_{k}^{A}(\vec{x}_{k}^{B})^{\mathrm{T}}\|_{\mathrm{tr}} = \sum_{k}p_{k}|\vec{x}_{k}^{A}||\vec{x}_{k}^{B}| \\
&\leq \max_{\vec{x}^{A}\in\mathcal{B}(\boldsymbol{X}^{A})}|\vec{x}^{A}| \max_{\vec{x}^{B}\in\mathcal{B}(\boldsymbol{X}^{B})}|\vec{x}^{B}| \; .
\end{align}
Here, we have used the convexity of norm for the first inequality and the boundedness of the convex set $\mathcal{B}(\boldsymbol{X})$ for the second inequality.
\end{proof}

\section{The proof of \cref{eq:ent_witness,eq:ent_witness_opt}}\label{appen:ent_witness}
\noindent
The Eqs. (56), (57) in the main text is
\begin{align}
&\mathcal{W} = \kappa\mathds{1} - \sum_{\mu}X_{\mu}^{A}\otimes X_{\mu}^{B} \; , \\
&\mathcal{W}_{O} = \kappa\mathds{1} - \sum_{\mu}\widetilde{X}_{\mu}^{A}\otimes \widetilde{X}_{\mu}^{B} \; .
\end{align}
Here, $\boldsymbol{\widetilde{X}}^{A}=(O^{A})^{\mathrm{T}}\boldsymbol{X}^{A}$, $\boldsymbol{\widetilde{X}}^{B}=(O^{B})^{\mathrm{T}}\boldsymbol{X}^{B}$ and $O^{A},O^{B}$ correspond to a singular value decomposition of $\mathcal{C}$, i.e. $\mathcal{C}=O^{A}\Sigma (O^{B})^{\mathrm{T}}$.
\begin{proof}
---Obviously, $\mathcal{W}$ is a Hermitian operator. Thus, we only need to prove $\tr[\mathcal{W}\rho_{\mathrm{sep}}]\geq 0$.
\begin{align}
\tr[\mathcal{W}\rho_{\mathrm{sep}}] &= \kappa - \sum_{\mu}\braket{X_{\mu}^{A}\otimes X_{\mu}^{B}} \notag \\
&= \kappa - \sum_{i}p_{i}\sum_{\mu}\braket{X_{\mu}^{A}}_{i}\braket{X_{\mu}^{B}}_{i} \notag \\
&= \kappa - \sum_{i}p_{i}\vec{x}_{i}^{A}\cdot\vec{x}_{i}^{B} \\
&\geq \kappa - \sum_{i}p_{i}|\vec{x}_{i}^{A}||\vec{x}_{i}^{B}| \geq 0 \; .
\end{align}
And then we prove that Theorem 2 is equivalent to the optimal entanglement witness $\mathcal{W}_{O}$.
\begin{align}
\tr[\mathcal{W}_{\sss O}\rho_{\mathrm{sep}}] &= \kappa - \sum_{\mu}\braket{\widetilde{X}_{\mu}^{A}\otimes \widetilde{X}_{\mu}^{B}} \notag \\
&= \kappa - \sum_{\mu}\Braket{\sum_{\mu'}(O^{A})^{\mathrm{T}}_{\mu\mu'}X_{\mu'}^{A}\otimes \sum_{\nu'}(O^{B})^{\mathrm{T}}_{\mu\nu'}X_{\nu'}^{B}} \\
&= \kappa - \sum_{\mu}\sum_{\mu'\nu'}(O^{A})^{\mathrm{T}}_{\mu\mu'}(O^{B})^{\mathrm{T}}_{\mu\nu'}\braket{X_{\mu'}^{A}\otimes X_{\nu'}^{B}} \\
&= \kappa - \sum_{\mu}\sum_{\mu'\nu'}(O^{A})^{\mathrm{T}}_{\mu\mu'}\mathcal{C}_{\mu'\nu'}O^{B}_{\nu'\mu} \\
&= \kappa - \tr[(O^{A})^{\mathrm{T}}\mathcal{C} O^{B}] \\
&= \kappa - \tr[\Sigma] \\
&= \kappa - \|\mathcal{C}\|_{\mathrm{tr}} \; .
\end{align}
Here, we make use of the singular value decomposition of $\mathcal{C}=O^{A}\Sigma (O^{B})^{\mathrm{T}}$.
\end{proof}

\section{The proof of $\max\sum_{\mu}\gamma_{\mu\mu}=\|\gamma\|_{\mathrm{tr}}$}\label{appen:gamma_orbit_opt}
\noindent
\textbf{von Neumann's trace theorem \cite{horn12_asdfd}.} Let the ordered singular values of $A,B\in M_{n}$ be $\sigma_{1}(A)\geq\cdots\geq\sigma_{n}(A)$ and $\sigma_{1}(B)\geq\cdots\geq\sigma_{n}(B)$. Then
\begin{align}
\operatorname{Re} \tr[AB] \leq \sum_{i}^{n}\sigma_{i}(A)\sigma_{i}(B) \; .
\end{align}
If $B=U$ and $UU^{\dagger}=\mathds{1}$, we reach the following corollary from von Neumann's trace theorem
\begin{align}
\operatorname{Re} \tr[AU] \leq \sum_{i}^{n}\sigma_{i}(A) = \|A\|_{\mathrm{tr}} \; ,
\end{align}
and the maximum can be reached with $U=WV^{\dagger}$, where $A=V\Sigma W^{\dagger}$ is a singular value decomposition of $A$. Thus, if $A$ is a real matrix, we have
\begin{align}
\max_{O\in O(n)}\tr[AO] = \|A\|_{\mathrm{tr}} \; .
\label{eq:neum_theore_coro}
\end{align}
Here, $O(n)$ is orthogonal group in dimension $n$. Thus
\begin{align}
\max_{\boldsymbol{X}^{A}\in\mathcal{O}(\boldsymbol{X}^A),\boldsymbol{X}^{B}\in\mathcal{O}(\boldsymbol{X}^B)}\sum_{\mu}\gamma_{\mu\mu} =& \max_{O^{A},O^{B}\in O(m)}\sum_{\mu}\Bigg[\Braket{\sum_{\mu'}O_{\mu\mu'}^{A}X_{\mu'}^{A}}\Braket{\sum_{\nu'}O_{\mu\nu'}^{B}X_{\nu'}^{B}} \\
&- \Braket{\sum_{\mu'}O_{\mu\mu'}^{A}X_{\mu'}^{A}\otimes \sum_{\nu'}O_{\mu\nu'}^{B}X_{\nu'}^{B}}\Bigg] \\
&= \max_{O^{A},O^{B}\in O(m)}\sum_{\mu}\sum_{\mu'\nu'}O_{\mu\mu'}^{A}O_{\mu\nu'}^{B}\left[\braket{X_{\mu'}^{A}}\braket{X_{\nu'}^{B}} - \braket{X_{\mu'}^{A}\otimes X_{\nu'}^{B}}\right] \\
&= \max_{O^{A},O^{B}\in O(m)}\sum_{\mu}\sum_{\mu'\nu'}O_{\mu\mu'}^{A}\gamma_{\mu'\nu'}(O^{B})^{\mathrm{T}}_{\nu'\mu} \\
&= \max_{O^{A},O^{B}\in O(m)}\tr[O^{A}\gamma (O^{B})^{\mathrm{T}}] \\
&= \max_{O\in O(m)}\tr[\gamma O] \\
&= \|\gamma\|_{\mathrm{tr}} \; .
\end{align}
Here, $O=(O^{B})^{\mathrm{T}}O^{A}$ and we make use of the corollary \cref{eq:neum_theore_coro} of von Neumann trace theorem.

\section{The proof of \cref{coro:sep_norm,coro:gene_sep}}\label{appen:coro_proof}
\begin{proof}
---Corollary 4: For the orthogonal measurement $\{X_{\mu}|\tr[X_{\mu}X_{\nu}]=\delta_{\mu\nu}\}$, $\sqrt{2}M_{A}$ and $\sqrt{2}M_{B}$ are orthogonal matrices and $\|\mathcal{C}\|_{\mathrm{tr}} = \|M_{A}^\mathrm{T}\chi M_{B}\|_{\mathrm{tr}}=\frac{1}{2}\|\chi\|_{\mathrm{tr}}$. Thus, CCNR criterion can be reformulated as $\|\chi\|_{\mathrm{tr}}\leq 2$. For normal form, we have
\begin{align}
\tilde{\rho} = \frac{1}{d^2}\mathds{1}\otimes\mathds{1} + \widetilde{\chi}_{\mu\nu}\pi_{\mu}\otimes\pi_{\nu} \; .
\end{align}
Thus, $\chi=\chi_{00}\oplus\widetilde{\chi}=\frac{2}{d}\oplus\widetilde{\chi}$ and $\|\chi\|_{\mathrm{tr}}=\|\frac{2}{d}\oplus\widetilde{\chi}\|_{\mathrm{tr}}=\frac{2}{d}+\|\tilde{\chi}\|_{\mathrm{tr}}\leq 2$, that is, $\|\widetilde{\chi}\|_{\mathrm{tr}} \leq \frac{2(d-1)}{d}$. For measurement $\{\frac{h}{\sqrt{d}}\mathds{1},\frac{\pi_{\mu}}{\sqrt{2}}\}$, we have the separability condition via Theorem 1
\begin{align}
\|\mathcal{C}\|_{\mathrm{tr}} = \|M_{A}^\mathrm{T}\chi M_{B}\|_{\mathrm{tr}} \leq \frac{d-1+h^2}{d} \; .
\end{align}
Here, $M_{A}=M_{B}=\frac{h}{\sqrt{2}}\oplus\frac{1}{\sqrt{2}}\mathds{1}$ and $\chi=\frac{2}{d}\oplus\widetilde{\chi}$ for normal form. Thus, $\|\mathcal{C}\|_{\mathrm{tr}}=\|\frac{h^2}{d}\oplus\frac{1}{2}\widetilde{\chi}\|_{\mathrm{tr}}=\frac{h^2}{d}+\frac{1}{2}\|\widetilde{\chi}\|_{\mathrm{tr}}\leq \frac{d-1+h^2}{d}$, that is, $\|\widetilde{\chi}\|_{\mathrm{tr}}\leq \frac{2(d-1)}{d}$. For measurement SCM, we have the separability condition via Theorem 1
\begin{align}
\|\mathcal{C}\|_{\mathrm{tr}} = \|M_{A}^\mathrm{T}\chi M_{B}\|_{\mathrm{tr}} \leq \frac{2d\alpha^2}{d+1}+h^2 \; .
\end{align}
Here,
\begin{align}
M_{A}=(\vec{n}_{1}^{A},\vec{n}_{2}^{A},\cdots,\vec{n}_{d^2}^{A})=\left(
\begin{matrix}
\frac{h}{\sqrt{2d}} & \frac{h}{\sqrt{2d}} & \cdots & \frac{h}{\sqrt{2d}} \\
\hat{n}_{1}^{A} & \hat{n}_{2}^{A} & \cdots & \hat{n}_{d^2}^{A}
\end{matrix}
\right) \; ,
\end{align}
and $\hat{n}_{1}^{A},\hat{n}_{2}^{A},\cdots,\hat{n}_{d^2}^{A}$ constitute a regular $(d^2-1)$-simplex, similarly for $M_{B}$. Thus, in light of \cref{eq:sup_simp_scm}, we have
\begin{align}
M_{A}M_{A}^{\mathrm{T}} = \left(
\begin{matrix}
\frac{dh^2}{2} & 0 \\
0 & \sum_{\mu}\hat{n}_{\mu}^{A}(\hat{n}_{\mu}^{A})^{\mathrm{T}}
\end{matrix}
\right) = \frac{dh^2}{2}\oplus\frac{d^2\alpha^2}{d^2-1} \mathds{1} \; .
\end{align}
Similarly, $M_{B}M_{B}^{\mathrm{T}}=\frac{dh^2}{2}\oplus\frac{d^2\alpha^2}{d^2-1} \mathds{1}$. Trace norm can be rewritten as $\|\mathcal{C}\|_{\mathrm{tr}}=\tr\left[\sqrt{\mathcal{C}\mathcal{C}^{\mathrm{T}}}\right]$, thus
\begin{align}
\|\mathcal{C}\|_{\mathrm{tr}}
&= \tr\left[\sqrt{\chi M_{B}M_{B}^{\mathrm{T}}\chi^{\mathrm{T}} M_{A}M_{A}^{\mathrm{T}}}\right] \\
&= \tr\left[\sqrt{h^4\oplus(\frac{d^2\alpha^2}{d^2-1})^2\widetilde{\chi}\widetilde{\chi}^{\mathrm{T}}}\right] \\
&= h^2 + \frac{d^2\alpha^2}{d^2-1}\|\widetilde{\chi}\|_{\mathrm{tr}} \; ,
\end{align}
where, $\chi=\frac{2}{d}\oplus\widetilde{\chi}$ for normal form. Finally, we obtain $\|\widetilde{\chi}\|_{\mathrm{tr}}\leq \frac{2(d-1)}{d}$.

\noindent
Corollary 5: It is very similar to the proof of Corollary 4. We only need to notice that the first column and row of $\chi'-\chi$ are zero, that is $\chi'-\chi=0\oplus \kappa$ and $\kappa_{\mu\nu}=\chi'_{\mu\nu}-\chi_{\mu\nu}$, $\mu,\nu=1,2,\cdots,d^2-1$.
\end{proof}

\section{The proof of \cref{th:univ_steering_cri}}\label{appen:univ_steering_cri}
\begin{proof}
---For a unsteerable state, the joint probability distribution always can be expressed
\begin{align}
P(a,b|X^{A},X^{B};\rho) = \sum_{\xi}\wp(a|X^{A},\xi)\tr[\Pi_{b}^{B}\rho_{\xi}]\wp_{\xi} \; .
\end{align}
So, the correlation matrix elements can be written as
\begin{align}
\mathcal{C}_{\mu\nu} &= \braket{X_{\mu}^{A}\otimes X_{\nu}^{B}} = \sum_{a_{\mu},b_{\nu}}a_{\mu}b_{\nu}P(a_{\mu},b_{\nu}|X_{\mu}^{A},X_{\nu}^{B};\rho) \notag \\
&= \sum_{a_{\mu},b_{\nu}}a_{\mu}b_{\nu}\left(\sum_{\xi}\wp(a_{\mu}|X_{\mu}^{A},\xi)\tr[\Pi_{b_{\nu}}^{B}\rho_{\xi}]\wp_{\xi}\right) \notag \\
&= \sum_{\xi}\left(\sum_{a_{\mu}}a_{\mu}\wp(a_{\mu}|X_{\mu}^{A},\xi)\right)\tr\left[\sum_{b_{\nu}}b_{\nu}\Pi_{b}^{B}\rho_{\xi}\right]\wp_{\xi} \notag \\
&= \sum_{\xi}\alpha_{\xi\mu}\tr[X_{\nu}^{B}\rho_{\xi}]\wp_{\xi} \notag \\
&= \sum_{\xi}\wp_{\xi}\alpha_{\xi\mu}x_{\xi\nu}^{B} \; .
\end{align}
Here, $\{a_{\mu}\},\{b_{\nu}\}$ are eigenvalues of $X_{\mu}^{A}$ and $X_{\nu}^{B}$ respectively and $\alpha_{\xi\mu}=\sum_{a_{\mu}}a_{\mu}\wp(a_{\mu}|X_{\mu}^{A},\xi)$, $x_{\xi\nu}^{B}=\tr[X_{\nu}^{B}\rho_{\xi}]$. Thus, the correlation matrix $\mathcal{C}$ is
\begin{align}
\mathcal{C} = \sum_{\xi}\wp_{\xi}\vec{\alpha}_{\xi}(\vec{x}_{\xi}^{B})^{\mathrm{T}} \; .
\end{align}
Here, $\vec{\alpha}_{\xi}=(\alpha_{\xi 1},\alpha_{\xi 2},\cdots)^{\mathrm{T}}$ and $\vec{x}_{\xi}^{B}=(\braket{X_{1}^{B}}_{\xi},\braket{X_{2}^{B}}_{\xi},\cdots)^{\mathrm{T}}$. In order to reach the final result, we need to calculate the bound of $|\vec{\alpha}_{\xi}|$,
\begin{align}
|\vec{\alpha}_{\xi}| &= \sqrt{\sum_{\mu}\alpha_{\xi\mu}^2} \\
&= \sqrt{\sum_{\mu}\left(\sum_{a_{\mu}}a_{\mu}\wp(a_{\mu}|X_{\mu}^{A},\xi)\right)^2} \\
&\leq \sqrt{\sum_{\mu}\sum_{a_{\mu}}a_{\mu}^2\wp(a_{\mu}|X_{\mu}^{A},\xi)} \\
&\leq \sqrt{\sum_{\mu}\max{\{a_{\mu}^2\}}} \\
&= \sqrt{\sum_{\mu}\lambda_{\mathrm{max}}((X_{\mu}^{A})^2)} \; .
\end{align}
Here, $\lambda_{\mathrm{max}}(X)$ refers to the maximal eigenvalue of $X$ and the first inequality is due to the nonnegativity of the variance. Thus
\begin{align}
\|\mathcal{C}\|_{\mathrm{tr}} &\leq \sum_{\xi}\wp_{\xi}\|\vec{\alpha}_{\xi}(\vec{x}_{\xi}^{B})^{\mathrm{T}}\|_{\mathrm{tr}} = \sum_{\xi}\wp_{\xi}|\vec{\alpha}_{\xi}||\vec{x}_{\xi}^{B}| \\
&\leq \sqrt{\sum_{\mu}\lambda_{\mathrm{max}}((X_{\mu}^{A})^2)} \max_{\vec{x}^{B}\in\mathcal{B}(\boldsymbol{X}^{B})}|\vec{x}^{B}| \; .
\end{align}
We have proved that $\|\mathcal{C}\|_{\mathrm{tr}}$ is a measurement orbit invariant in \cref{ob:measu_orbit_invar}. So, the optimal result is to optimize the $\sqrt{\sum_{\mu}\lambda_{\mathrm{max}}((X_{\mu}^{A})^2)}$ in the measurement orbit, that is, $\displaystyle\min_{\boldsymbol{X}^{A}\in\mathcal{O}(\boldsymbol{X}^A)}\sqrt{\sum_{\mu}\lambda_{\mathrm{max}}((X_{\mu}^{A})^2)}$. Finally, we have
\begin{align}
\|\mathcal{C}\|_{\mathrm{tr}} \leq \min_{\boldsymbol{X}^{A}\in\mathcal{O}(\boldsymbol{X}^A)}\sqrt{\sum_{\mu}\lambda_{\mathrm{max}}((X_{\mu}^{A})^2)}\max_{\vec{x}^{B}\in\mathcal{B}(\boldsymbol{X}^{B})}|\vec{x}^{B}| \; .
\end{align}
\end{proof}

% \bibliographystyle{ref_style}
% \bibliography{ref.bib}

\end{document}